\documentclass[10pt,journal,final,letterpaper,oneside,twocolumn]{IEEEtran}
\ifCLASSINFOpdf
\else
\fi
\usepackage{amsthm}
\usepackage{amssymb}
\usepackage{graphicx}
\usepackage{amsmath}
\usepackage{bbold}
\usepackage{bbm}
\usepackage{nonfloat}
\usepackage{color}
\usepackage[pdftex]{hyperref}
\usepackage{braket}
\usepackage{graphicx}
\usepackage{dsfont}
\usepackage{mathdots}
\usepackage{mathtools}

\newtheorem{thm}{Theorem}
\newtheorem*{thm*}{Theorem}
\newtheorem{prop}[thm]{Proposition}
\newtheorem*{prop*}{Proposition}
\newtheorem{lemma}[thm]{Lemma}
\newtheorem*{lemma*}{Lemma}
\newtheorem{cor}[thm]{Corollary}
\newtheorem*{cor*}{Corollary}

\newtheorem*{ex*}{Example}

\newtheorem*{cj*}{Conjecture}

\newtheorem*{Def*}{Definition}

\theoremstyle{definition}
\newtheorem*{rem}{Remark}

\newcommand{\bq}{\begin{equation*}}
\newcommand{\be}{\begin{equation}}
\newcommand{\eq}{\end{equation*}}
\newcommand{\ee}{\end{equation}}
\newcommand{\lmatrix}{\left(\begin{smallmatrix}}
\newcommand{\rmatrix}{\end{smallmatrix}\right)}

\newcommand{\Tr}{\operatorname{Tr}}
\newcommand{\rk}{\operatorname{rank}}
\newcommand*{\coloneqq}{\mathrel{\vcenter{\baselineskip0.5ex \lineskiplimit0pt \hbox{\scriptsize.}\hbox{\scriptsize.}}} =}
\newcommand{\texteq}[1]{\stackrel{\mathclap{\scriptsize \mbox{#1}}}{=}}
\newcommand{\textleq}[1]{\stackrel{\mathclap{\scriptsize \mbox{#1}}}{\leq}}
\newcommand{\textgeq}[1]{\stackrel{\mathclap{\scriptsize \mbox{#1}}}{\geq}}

\begin{document}
%
\title{From log-determinant inequalities \\ to Gaussian entanglement via recoverability theory}
%
%
%

\author{Ludovico Lami, Christoph Hirche, Gerardo Adesso and Andreas Winter%
\thanks{Ludovico Lami, Christoph Hirche and Andreas Winter are with the
Departament de F\'{i}sica: Grup d'Informaci\'o Qu\`{a}ntica,
Universitat Aut\`{o}noma de Barcelona, ES-08193 Bellaterra (Barcelona), Spain.
Andreas Winter is furthermore with ICREA---Instituci\'o Catalana de
Recerca i Estudis Avan\c{c}ats, Pg. Lluis Companys, 23, ES-08010 Barcelona, Spain.}%
\thanks{Gerardo Adesso is with the Centre for the Mathematics and Theoretical Physics
of Quantum Non-Equilibrium Systems, School of Mathematical Sciences,
The University of Nottingham, University Park, Nottingham NG7 2RD, United Kingdom.}%
\thanks{LL, CH and AW acknowledge support from the Spanish MINECO, projects
FIS2013-40627-P and FIS2016-86681-P, with the support of FEDER funds;
from the Generalitat de Catalunya, CIRIT project 2014-SGR-966; and
from the European Research Council, Advanced Grant IRQUAT (2010-AdG-267386).
CH in addition is supported by FPI scholarship no.~BES-2014-068888.
GA acknowledges support from The European Research Council, Starting Grant
GQCOP (grant no.~637352) and the Foundational Questions Institute (FQXi)
``Physics of the Observer Programme'' (grant no.~FQXi-RFP-1601)}.%
\thanks{Copyright (c) 2017 IEEE}}

\maketitle

\begin{abstract}
Many determinantal inequalities for positive definite block matrices are consequences of general entropy inequalities, specialised to Gaussian distributed vectors with prescribed covariances.
In particular, strong subadditivity (SSA) yields
\begin{equation*}
\ln\det V_{AC} + \ln\det V_{BC} - \ln\det V_{ABC} - \ln\det V_C \geq 0
\end{equation*}
for all $3\times 3$-block matrices $V_{ABC}$, where subscripts identify principal submatrices. We shall refer to the above inequality as SSA of log-det entropy. In this paper we develop further insights on the properties of the above inequality and its applications to classical and quantum information theory.

In the first part of the paper, we show how to find known and new necessary and sufficient conditions under which saturation with equality occurs. Subsequently, we discuss the role of the classical transpose channel (also known as Petz recovery map) in this problem and find its action explicitly. We then prove some extensions of the saturation theorem, by finding faithful lower bounds on a log-det conditional mutual information.

In the second part, we focus on quantum Gaussian states, whose covariance matrices are not only positive but obey additional constraints due to the uncertainty relation. For Gaussian states, the log-det entropy is equivalent to the R\'enyi entropy of order $2$. We provide a strengthening of log-det SSA for quantum covariance matrices that involves the so-called Gaussian \mbox{R\'enyi-$2$} entanglement of formation, a well-behaved entanglement measure defined via a Gaussian convex roof construction. We then employ this result to define a log-det entropy equivalent of the squashed entanglement measure, which is remarkably shown to coincide with the Gaussian \mbox{R\'enyi-$2$} entanglement of formation. This allows us to establish useful properties of such measure(s), like monogamy, faithfulness, and additivity on Gaussian states.
\end{abstract}


%
\IEEEpeerreviewmaketitle

\section{Introduction}
%
%
%
%

\label{sec:intro}

The idea of using information theoretical reasoning to prove
determinantal inequalities for positive definite matrices has been the
subject of growing interest in the last decades (see e.g.~the reviews
given in~\cite{Cover,Dembo}). The key of the above correspondence is to associate, to each positive
matrix\footnote{In this paper we consider only real matrices since they are more relevant for the applications we are interested in, but all the results we find apply also to the Hermitian case with minor modifications. } $A\in\mathcal{M}_n(\mathds{R})$, an $n$-dimensional Gaussian random variable
$X\in\mathds{R}^n$ with mean $0$ and variance (aka covariance matrix)
$\operatorname{Var} X = \mathds{E}\  X X^\intercal = A$.
The density of $X$ is given by
\begin{equation}
  p_A(x) = \frac{e^{-\frac12 x^\intercal A^{-1} x}}{\sqrt{(2\pi)^n\det A}} .
  \label{gauss}
\end{equation}
This has the nice feature that for two independent Gaussian random
variables $X$ and $Y$ with mean $0$ and covariance matrices $A$ and $B$, respectively,
the sum $A+B$ is the covariance matrix of $X+Y$.

Under the density~\eqref{gauss}, the differential entropy
$h(X) \coloneqq -\int d^n x\, p_A(x) \ln p_A(x)$ of~\eqref{gauss} takes the form
\begin{equation}
  h(X) = \frac{1}{2}\ln\det A + \frac{n}{2}\left(\ln 2\pi + 1 \right) ,
  \label{ent}
\end{equation}
while the relative entropy
$D(p_A\|p_B) \coloneqq \int d^n x\, p_A(x) \ln \frac{p_A(x)}{p_B(x)}$ is given by
\begin{equation}
  D(p_A\|p_B) = \frac{1}{2} \ln\frac{\det B}{\det A} + \frac{1}{2} \Tr(B^{-1}\! A) - \frac{n}{2} .
  \label{rel ent}
\end{equation}
Here and in the remainder of the paper we denote by $\ln$ the natural logarithm. The positivity of~\eqref{rel ent} as a function of the matrices $A$ and $B$ can be
seen as an instance of Klein's inequality applied to the natural logarithm~\cite{Wehrl}.

In this picture, general inequalities involving entropies can be turned into
inequalities involving determinants thanks to~\eqref{ent} and~\eqref{rel ent}.
A prominent example of the usefulness of this approach is constituted by
\emph{strong subadditivity} (SSA), the basic ``Shannon-type'' entropy
inequality~\cite{Yeung}.
Consider a Gaussian distributed vector
$X_{ABC} = (X_A,X_B,X_C)^\intercal \in \mathds{R}^{n_A+n_B+n_C}$
with covariance matrix $V_{ABC}$:
\begin{equation}
  V_{ABC} = \begin{pmatrix} A & X & Y \\
                            X^{\intercal} & B & Z \\
                            Y^{\intercal} & Z^{\intercal} & C \end{pmatrix}
    \geq 0 .
  \label{global CM}
\end{equation}
The SSA inequality $I(X_{A}:X_{B}|X_{C}) \geq 0$ then reads
\begin{equation}
  \ln\det V_{AC} + \ln\det V_{BC} - \ln\det V_{ABC} - \ln\det V_C \geq 0 ,
  \label{SSA}
\end{equation}
where the local reductions $V_{AC}$, $V_{BC}$ and $V_{C}$ are the principal
submatrices of $V_{ABC}$ corresponding to the components $AC$, $BC$ and $C$,
respectively:
\begin{equation}
  V_{AC} = \begin{pmatrix} A & Y \\ Y^{\intercal} & C \end{pmatrix}, \quad
  V_{BC} = \begin{pmatrix} B & Z \\ Z^{\intercal} & C \end{pmatrix}, \quad
  V_C = C.
  \label{eq:reductions}
\end{equation}
Let us observe that since~\eqref{SSA} is balanced, the
contribution of the inhomogeneous second terms of~\eqref{ent} cancel out.

Inequality~\eqref{SSA} was considered in~\cite{Ando09} (see also~\cite[Sec.~4.5]{Petz book}), although it has been known long before under the name of Hadamard-Fisher inequality.
From the point of view of matrix analysis,
\eqref{SSA} lends itself to straightforward generalisations. In fact,
inequalities of the same form have recently been investigated. In
particular, the problem of determining all the continuous functions
$f:\mathds{R}_{+}\rightarrow\mathds{R}$ such that
for all block matrices $V_{ABC} \geq 0$,
\begin{equation}
  \Tr f(V_{AC}) + \Tr f(V_{BC}) - \Tr f(V_{ABC}) - \Tr f(C) \geq 0 ,
  \label{SSA gen}
\end{equation}
was considered in full generality in~\cite{Audenaert10}, where a sufficient
condition was found:~\eqref{SSA gen} holds as soon as $-f'$ is matrix monotone.
Later on, it was shown that this condition is also necessary~\cite{Lewin14}.
By virtue of L\"owner's theorem characterising matrix monotone
functions~\cite{Bhatia analysis}, this yields an explicit
characterisation of all the functions $f$ obeying~\eqref{SSA gen}. Here we
are mainly concerned with the particular choice $f(x)=\ln x$, that
turns~\eqref{SSA gen} into~\eqref{SSA}.
Incidentally, the differential R\'enyi-$\alpha$ entropy of a Gaussian random variable $X$ with density
$p_A(x)$, i.e.~$H_\alpha(X) \coloneqq \frac{1}{1-\alpha}\ln \int d^n x\, p_A(x)^\alpha$,
is given by
\begin{equation*}
  h_\alpha(X) = \frac{1}{2} \ln\det A + \frac{n}{2}\left( \ln 2\pi + \frac{1}{\alpha-1}\ln\alpha \right),
\end{equation*}
showing that all the differential R\'enyi entropies of Gaussian random vectors are
essentially equivalent to the differential Shannon entropy, up to a characteristic
universal additive offset.
In view of this and the above remarks, we are motivated,
given a vector valued random variable $X$ with covariance
matrix $V$, to refer from now on to the quantity
\begin{equation}
  \label{logdetent}
  M(X) \coloneqq M(V) \coloneqq \frac12 \ln \det V ,
\end{equation}
as the {\it log-det entropy} of $V$.
Likewise, for a bipartite covariance matrix $V_{AB} > 0$ we refer to
\begin{equation}\begin{split}
  I_{M}(A:B)_{V} &\coloneqq \frac12 \ln\frac{\det V_{A}\det V_{B}}{\det V_{AB}} \\
                   &=      M(V_{A})+M(V_{B})-M(V_{AB}) ,
  \label{logdetMI}
\end{split}\end{equation}
as the {\em log-det mutual information}, and for a tripartite covariance matrix $V_{ABC}>0$ we refer to
\begin{equation}\begin{split}
  I_{M}(A:B|C)_{V} &\coloneqq \frac12 \ln\frac{\det V_{AC}\det V_{BC}}{\det V_{C}\det V_{ABC}} \\
                   &=      M(V_{AC})+M(V_{BC}) \\
                   &\quad -M(V_{ABC})-M(V_{C}) ,
  \label{I_2-matrix}
\end{split}\end{equation}
as the {\em log-det conditional mutual information}.


Every (balanced) entropic inequality thus yields a corresponding log-determinant
inequality for positive block matrices~\cite{Chan-balanced}. Thanks
to the work of Zhang and Yeung~\cite{ZhangYeung}
and followers~\cite{Dougherty,Matus}, infinitely many
independent such inequalities, so-called ``non-Shannon-type inequalities'',
are known by now.
The question of what are the precise
constraints on the determinants of the $2^n$ principal
submatrices of a positive matrix of size $n\times n$ has been raised much earlier,
either directly in a matrix setting~\cite{JohnsonBarrett}
or more recently in the guise of the balanced entropy inequalities of Gaussian
random variables (both real valued or vector valued)~\cite{Hassibi,Shadbakht}.
Remarkably, the latter papers show that while the entropy region of three Gaussian
real random variables is convex but not a cone, the entropy region of three Gaussian
random vectors is a convex cone and that the \emph{linear} log-det inequalities
for three Gaussian random variables (and equivalently Gaussian random vectors)
are the same as the inequalities for the differential entropy of any three
variables -- which in turn coincide with the Shannon inequalities,
cf.~\cite{Yeung,Chan-balanced}.
It is conjectured that the same identity between Gaussian vector inequalities
and general differential inequalities holds for any number parties.

In the present paper, we will focus on a deeper investigation of the
SSA inequality~\eqref{SSA}. Our analysis rests crucially on the connection between Gaussian random variables and positive definite matrices we have outlined here, which allows us to use tools taken from matrix analysis~\cite{Bhatia positive} to explore properties of the log-det conditional mutual information~\eqref{I_2-matrix}.
This route has been already undertaken in our recent work~\cite{Lami16}, in which we have
shown that the inequality~\eqref{SSA} can be strengthened significantly to the following
matrix inequality (with respect to the semidefinite, or L\"owner,
ordering on symmetric matrices):
\begin{equation}
  V_{ABC} / V_{BC} \leq V_{AC} / V_C ,
  \label{INEQ 1}
\end{equation}
using the powerful concept of \emph{Schur complement} of a $2\times 2$-block matrix
$V = \lmatrix A & X \\ X^\intercal & B \rmatrix$ with respect to the
principal minor $A$, denoted as
\begin{equation}
\label{Schur definition}V/A\coloneqq B - X^{\intercal}A^{-1} X .
\end{equation}
We will go into more detail about the properties of the Schur complement in the next section.

Our concrete interest in~\eqref{SSA} is partly motivated by its applications
in quantum information theory with continuous
variables~\cite{Adesso14}, as first explored in~\cite{Adesso12,Gross}.
Every continuous variable quantum state $\rho$ of $n$ modes, subject to mild regularity
conditions, has a $2n\times 2n$-covariance matrix $V$ of the phase space variables.
By slight abuse of terminology, we shall call $M(V) = \frac12 \ln\det V$
the {\it log-det entropy} of $\rho$, and denote it equivalently as $M(\rho)$,
\begin{equation}
  M(\rho) \coloneqq M(V) \coloneqq \frac12 \ln \det V.
\end{equation}
Analogously, quantities like the log-det conditional mutual information can be defined for an arbitrary (sufficiently regular) state via its covariance matrix, i.e.
\begin{equation}
\begin{split}
  I_{M}(A:B|C)_{\rho} &\coloneqq I_{M}(A:B|C)_{V} \\
  &= \frac12 \ln\frac{\det V_{AC}\det V_{BC}}{\det V_{C}\det V_{ABC}}
  \label{I_2}
\end{split}
\end{equation}
where $\rho_{ABC}$ is a tripartite state and $V_{ABC}$ its covariance matrix. Thus, by construction the log-det conditional mutual information quantifies correlations encoded in the
second moments of the state. Observe how the above combination of log-det entropies
mimics that appearing in the celebrated SSA of the quantum von Neumann entropy
$S$~\cite{Robinson67,Lanford68,Lieb73},
\begin{equation}
  \label{SSSA}
  S(\rho_{AC})+S(\rho_{BC})-S(\rho_{C})-S(\rho_{ABC}) \geq 0 ,
\end{equation}
which is nowadays widely regarded as one of the cornerstones upon which quantum information theory is built~\cite{Nielsen}.

Remarkably, in the particular case of interest in which $\rho$ is a quantum Gaussian
state, that is, a state with Gaussian phase space Wigner function~\cite{Adesso14},
the log-det entropy reduces to the quantum \mbox{R\'enyi-$2$} entropy $S_2$ of $\rho$,
\begin{equation}
  S_{2}(\rho) \coloneqq -\ln \Tr \rho^{2} = \frac12 \ln \det V = M(V) .
\end{equation}
Therefore, in the relevant case of tripartite quantum Gaussian states, the
general inequality~\eqref{SSA} for log-det entropy takes the form of a
SSA inequality for the \mbox{R\'enyi-$2$} entropy~\cite{Adesso12,Gross,Adesso}, holding in addition to the standard
one for R\'enyi-1 entropy aka von Neumann entropy, which is valid for
arbitrary (Gaussian or not) tripartite quantum states.

The usefulness of inequalities like~\eqref{SSA} in quantum optics and
quantum information was acknowledged in a series of recent papers.
In~\cite{Adesso} (see also~\cite{Kor}) it was proven that an alternative (non-balanced) formulation
of~\eqref{SSA}, obtained via a conventional purification procedure, leads
to a remarkable limitation on the quantum steerability of tripartite
states via Gaussian measurements. Namely, it is not possible for a
single-mode system to be steered simultaneously by two multimode parties
via Gaussian measurements.
As one could expect, operator
inequalities like~\eqref{INEQ 1} have even stronger implications for
quantum correlations in tripartite systems, leading for instance to a
fundamental monogamy constraint on the \mbox{R\'enyi-$2$} Gaussian entanglement
of formation~\cite{Lami16}.

\medskip
The rest of the present paper is structured as follows. In Section~\ref{sec exact recov}
we derive various characterisations of the case of saturation of SSA~\eqref{SSA} with
equality.
Then, in Section~\ref{sec:Gaussian-recov} we turn to the case of near-saturation,
which leads to the theory of recovery maps; in Section~\ref{sec:I_M-lower-bound}
we exploit those results to derive simple and faithful lower bounds on
the log-det conditional mutual information. Up to that point, all results
hold for general covariance matrices $V>0$.
After that, in Section~\ref{sec:Renyi-2-Gaussian-squashed} we turn our
attention to quantum Gaussian states and their phase space covariance
matrices, which need to satisfy additional constraints stemming from the uncertainty principle and the canonical commutation relations. There, we introduce a measure of entanglement for quantum Gaussian states
based on the log-det conditional mutual information defined in~\eqref{I_2}
and prove its faithfulness and additivity. Quite remarkably, we show that the measure coincides with the \mbox{R\'enyi-$2$} Gaussian entanglement of formation introduced in~\cite{Adesso12}, equipping the latter with an interesting operational interpretation in the context of recoverability.
We conclude in Section~\ref{sec:conchi} with a number of open questions.

\section{Mathematical tools: Schur complement and geometric mean}

Two of the elementary tools we will use in the remainder of this paper are the Schur complement and the geometric mean between positive definite matrices. In this section we will state some useful properties and observations.

Let's start with the Schur complement~\cite{Schur}. First we recall its definition: given a $2\times 2$-block matrix
$V = \lmatrix A & X \\ X^\intercal & B \rmatrix$, the complement with respect to the
principal minor $A$ is given by $V/A$ as defined in~\eqref{Schur definition}.

Its significance relies on the (elementary) fact that $V$ as a quadratic form is congruent to
$S^\intercal V S = A \oplus V/A$, via the unideterminantal transformation
$S = \lmatrix \mathds{1} & -A^{-1}X \\ 0 & \mathds{1} \rmatrix$. From this
the factorisation formula
\begin{equation}
  \det V = (\det A)(\det V/A)
  \label{det factor}
\end{equation}
follows, which shows how~\eqref{INEQ 1} implies the SSA inequality~\eqref{SSA}.
From a point of view of linear algebra, Schur complements arise naturally when one wants to express the inverse of a block matrix in a compact form. Namely, for a matrix $V$ partitioned as above one can prove the useful formula~\cite{SCbook}
\begin{equation}
  V^{-1} = \lmatrix A^{-1} + A^{-1} X (V/A)^{-1} X^\intercal  A^{-1} & -A^{-1} X (V/A)^{-1} \\[0.7ex]
                                            -(V/A)^{-1} X^\intercal  A^{-1} & (V/A)^{-1} \rmatrix .
  \label{inv}
\end{equation}
Naturally, an analogous expression holds with $A$ and $B$ interchanged.
Incidentally, from this latter fact many useful matrix identities can be easily derived.

Schur complements of positive definite matrices enjoy numerous other
useful relations. First of all, the positivity condition itself can be expressed in terms of Schur
complements as
\begin{equation}
V = \begin{pmatrix} A & X \\ X^\intercal & B \end{pmatrix} > 0\quad \Longleftrightarrow\quad A>0\ \text{and}\ V/A>0 .
\label{Schur pos}
\end{equation}
From this the variational representation
\begin{equation}
V/A = \max \big\{ \tilde{B}: V \geq 0 \oplus \tilde{B} \big\} ,
\label{Schur var}
\end{equation}
follows easily. The meaning of~\eqref{Schur var} is that the matrix set on the right hand side has a unique maximal element with respect to the L\"owner partial order (a nontrivial fact in itself) and that this maximum coincides with the left hand side.
Another useful property is the additivity of ranks under Schur complements:
\begin{equation}
\rk V = \rk A + \rk (V/A) .
\label{rank add}
\end{equation}
We shall make use of these properties in the sequel. For more details on Schur complements and applications thereof in matrix analysis
and beyond we refer the reader to the book~\cite{SCbook}.

Another fundamental tool we shall take from matrix analysis is the concept of \emph{geometric mean} between two positive definite matrices $A,B>0$, usually denoted by $A\# B$~\cite{geom original, Ando79}. As done in~\eqref{Schur var} for the Schur complement, also the geometric mean is most conveniently defined using a variational approach. Namely, one has
\begin{equation}
A\# B \coloneqq \max\{X=X^{\intercal}:\, A\geq XB^{-1} X\}\, , \label{geom var}
\end{equation}
the maximum being taken with respect to the semidefinite order. From~\eqref{geom var} it is apparent, how $A\# B$ is covariant with respect to matrix congruence, i.e.
\begin{equation}
\left(SAS^\intercal \right)\# \left( SBS^\intercal \right) = S (A\# B) S^\intercal
\label{geom cov congr}
\end{equation}
for all invertible $S$. Moreover, through standard algebraic manipulations it is possible to write the explicit solution of~\eqref{geom var} as
\begin{equation}
A\# B = A^{1/2} \left( A^{-1/2} B A^{-1/2} \right)^{1/2} A^{1/2} . \label{geom expl}
\end{equation}
An excellent introduction to the theory of matrix means can be found in~\cite[Chapter 4]{Bhatia positive}. Here, we limit ourselves to briefly discuss an interesting interpretation of the geometric mean. We can turn the manifold of positive definite matrices into a Riemannian manifold by introducing on the tangent space the metric $ds^2\coloneqq \Tr[(A^{-1} dA)^2]$ (sometimes called ``\emph{trace metric}"). It turns out the geodesic connecting two positive matrices $A$ and $B$ in this metric, parametrised by $t\in [0,1]$, is given by
\begin{equation}
\gamma(t) = A^{1/2} \left( A^{-1/2} B A^{-1/2} \right)^{t} A^{1/2} \eqqcolon A\#_t B,
\label{geom geod}
\end{equation}
sometimes called the \emph{weighted geometric mean}. From this we see in particular that $A\# B$ is nothing but the geodesic midpoint between $A$ and $B$. An easy consequence of the above expression is the determinantal identity
\begin{equation}
\det (A\#_t B) = (\det A)^{1-t} (\det B)^t .
\label{det geom}
\end{equation}
For more on this connection between geometric mean and Riemannian metric, see~\cite[Chapter 6]{Bhatia positive}.



\section{SSA saturation and exact recovery}
\label{sec exact recov}
Now we turn to studying the conditions under which~\eqref{SSA}
is saturated with equality. A necessary and sufficient condition was already found
in~\cite{Ando09} (for a comprehensive discussion, see~\cite{Petz book}),
but here we present new proofs as well as alternative formulations,
which may provide new insights.

Let us start by fixing our notation concerning classical Gaussian channels,
whose action can be described as follows. Denote the input random variable
by $X$, and consider an independent Gaussian variable $Z \sim P_K$, where
$P_{K}$ is a normal distribution with covariance matrix $K$ and zero mean.
Then the output variable $Y$ of the Gaussian channel $N$ is given by
$N(X) \coloneqq Y \coloneqq H X + Z$ for some matrix $H$ of appropriate size. At the level of covariance matrices this
translates to
\begin{equation}
  N : V \longmapsto V' = H V H^\intercal + K ,
  \label{N}
\end{equation}
where the only constraint to be obeyed is $K\geq 0$.

The following theorem gathers some notable facts concerning log-det conditional mutual information, and provides a neat example of how useful the interplay between matrix analysis
and information theory with Gaussian random variables can be. We are going to employ these results extensively throughout the paper, and some of them play an important role already in the proof of the main theorem of this section.

\begin{thm}
  \label{thm I cond geom}
  For all positive, tripartite matrices $V=V_{ABC}>0$, the following identities hold true:
  \begin{align}
    I_{M}(A:B|C)_{V} &= I_{M}(A:B)_{V_{ABC}/V_C} ,     \label{I cond Schur}\\
    I_{M}(A:B|C)_{V} &= I_{M}(A:B)_{V^{-1}} .                \label{I cond inv}
  \end{align}
  Furthermore, for all pairs of positive definite matrices $V_{AB},W_{AB}>0$, the log-det mutual information is convex on the geodesic connecting
  them as in~\eqref{geom geod}, i.e.
  \begin{equation}
  I_M(A:B)_{V\#_t W} \leq (1-t) I_M(A:B)_V + t I_M(A:B)_W . \label{I conv geod}
  \end{equation}
\end{thm}

\begin{proof}
Let us start by showing~\eqref{I cond Schur}. Using repeatedly the determinant factorisation property~\eqref{det factor}, we find
\begin{align*}
    & I_{M}(A:B)_{V_{ABC}/V_C} \\[0.8ex]
    &\quad = \frac12 \ln \frac{\det (V_{AB}/V_C) \det (V_{BC}/V_C)}{\det (V_{ABC}/V_C)} \\[0.8ex]
    &\quad = \frac12 \ln \frac{(\det V_{AB})(\det V_C)^{-1} (\det V_{BC})(\det V_C)^{-1}}{(\det V_{ABC})(\det V_C)^{-1}} \\[0.8ex]
    &\quad= \frac12 \ln \frac{(\det V_{AB}) (\det V_{BC} )}{(\det V_{ABC})(\det V_C)} \\[0.8ex]
    &\quad= I_M(A:B|C)_V .
\end{align*}

We now move to~\eqref{I cond inv}. The block inverse formulae~\eqref{inv} give us
\begin{align*}
  (V^{-1})_{AB} &= (V_{ABC}/V_{C})^{-1} ,  \\
  (V^{-1})_{A}  &= (V_{ABC}/V_{BC})^{-1} , \\
  (V^{-1})_{B}  &= (V_{ABC}/V_{AC})^{-1} .
\end{align*}
Putting all together we find
\begin{align*}
    & I_{M}(A:B)_{V^{-1}} \\[0.8ex]
    &\quad = \frac12 \ln\frac{\det (V^{-1})_{A}\det (V^{-1})_{B}}{\det (V^{-1})_{AB}} \\[0.8ex]
    &\quad = \frac12 \ln\frac{\det (V_{ABC}/V_{BC})^{-1} \det (V_{ABC}/V_{AC})^{-1}}{\det (V_{ABC}/V_{C})^{-1}} \\[0.8ex]
    &\quad = \frac12 \ln\frac{\det (V_{ABC}/V_{C})}{\det (V_{ABC}/V_{BC}) \det (V_{ABC}/V_{AC})} \\[0.8ex]
    &\quad = \frac12 \ln\frac{(\det V_{ABC})(\det V_{C})^{-1}}
                        {(\det V_{ABC})(\det V_{BC})^{-1} (\det V_{ABC})(\det V_{AC})^{-1}} \\[0.8ex]
    &\quad = \frac12 \ln \frac{\det V_{AC}\det V_{BC}}{\det V_{ABC}\det V_{C}} \\[0.8ex]
    &\quad = I_{M}(A:B|C)_V ,
\end{align*}
which is what we wanted to show.

Finally, let us consider~\eqref{I conv geod}. A preliminary observation uses the monotonicity of the geometric mean under positive maps~\cite[Theorem~3]{Ando79}, written as $\Phi(V\# W)\leq \Phi(V)\# \Phi(W)$. Iterative applications of this inequality show that the same monotonicity property holds also for the weighted geometric mean~\eqref{geom geod} when $t$ is a dyadic rational, and hence (by continuity) for all $t\in [0,1]$. This standard reasoning is totally analogous to the one normally used to show that mid-point convexity and convexity are equivalent for continuous functions. Applying this to the positive map $\Phi(X)\coloneqq \Pi_A X \Pi_A^{\intercal}$, where $\Pi_A$ is the projector onto the $A$ components, yields $(V\#_t W)_A = \Pi_A (V\#_t W) \Pi_A^\intercal \leq V_A \#_t W_A$. Taking the determinant of both sides of this latter inequality and using for the right hand side the explicit formula~\eqref{det geom} we obtain $\det \left(V\#_t W \right)_A \leq \det\left(V_A \#_t W_A\right) = (\det V_A)^{1-t} (\det W_A)^t$. Together with the analogous inequality for the $B$ system, this gives
\begin{align*}
   & I_M(A:B)_{V\#_t W} \\[0.8ex]
   &\quad = \frac12 \ln \frac{\left(\det (V\#_t W)_A \right) \left( \det (V \#_t W)_B \right)}{\det (V\#_t W)_{AB}} \\[0.8ex]
   &\quad \leq \frac12 \ln \frac{(\det V_A)^{1-t} (\det W_A)^t (\det  V_B)^{1-t}  (\det W_B)^t}{(\det V_{AB})^{1-t} (\det W_{AB})^t} \\[0.8ex]
   &\quad = (1-t) I_M(A:B)_V + t I_M(A:B)_W ,
\end{align*}
concluding the proof.
\end{proof}

\begin{rem}
Inequality~\eqref{I conv geod} is especially notable because in general the log-det mutual information is not convex over the set of positive matrices. However, it is convex when restricted to geodesics in the trace metric, as we have just shown. Moreover, we note in passing that an analogous inequality to~\eqref{I conv geod} does not seem to hold for the log-det conditional mutual information.
\end{rem}

\begin{thm}
\label{thm satur}
  For an arbitrary $V_{ABC}>0$ written in block form as in~\eqref{global CM}, the following are equivalent:
  \begin{enumerate}
    \item $I_{M}(A:B|C)_{V}=0$, i.e.~\eqref{SSA} is saturated;
    \item $V_{ABC}/V_{BC}=V_{AC}/V_{C}$, i.e.~\eqref{INEQ 1} is saturated;
    \item $(V^{-1})_{AB}=(V^{-1})_{A}\oplus (V^{-1})_{B}$;
    \item $X=YC^{-1}Z^{\intercal}$ (see~\cite{Ando09} or~\cite[Thm.~4.49]{Petz book});
    \item there is a classical Gaussian channel $N_{C\rightarrow BC}$ such
          that $(I_A\oplus N_{C\rightarrow BC})(V_{AC})=V_{ABC}$.
  \end{enumerate}
\end{thm}

\begin{proof} $\\[-3ex]$
\begin{description}
\item[$1\!\Leftrightarrow\! 2$.] Saturation of~\eqref{SSA} and~\eqref{INEQ 1} are equivalent concepts, since it is very easy to verify that if $M\geq N>0$ then $M=N$ if and only if $\det M=\det N$.

\item[$1\!\Leftrightarrow\! 3$.] It is well-known that $W_{AB}>0$ satisfies $\det W_{AB} = \det W_{A} \det W_{B}$ iff its off-diagonal block is zero, i.e. iff $W_{AB}=W_{A}\oplus W_{B}$. For instance, this can be easily seen as a consequence of~\eqref{det factor}. Thanks to Theorem~\ref{thm I cond geom}, identity~\eqref{I cond inv}, applying this observation with $W=V^{-1}$ yields the claim.

\item[$2\!\Rightarrow\! 4$.] This is known in linear algebra~\cite{Ando09}, but for
the sake of completeness we provide a different proof that fits more with the
spirit of the present work. Namely, we see that the variational representation of Schur complements~\eqref{Schur var} guarantees that~\eqref{INEQ 1} is saturated if and only if
\begin{equation}
\begin{split}
  V_{ABC} - (V_{AC}/V_C) \oplus 0_{BC}\
    &= \lmatrix A-V_{AC}/V_C & X & Y \\
                      X^\intercal  & B & Z \\
                      Y^\intercal  & Z^\intercal  & C \rmatrix \\[0.8ex]
    &= \lmatrix YC^{-1}Y^\intercal  & X & Y \\
                      X^\intercal  & B & Z \\
                      Y^\intercal  & Z^\intercal  & C \rmatrix \\[0.8ex]
    &\geq 0\, .
\end{split}
  \label{eq cond eq1}
\end{equation}
A necessary condition for~\eqref{eq cond eq1} to hold is obtained by taking suitable matrix elements:
\begin{equation*}
\begin{split}
  0 &\leq \lmatrix v \\ w \\ -C^{-1}Y^\intercal  v \rmatrix^{\intercal}
\lmatrix YC^{-1}Y^\intercal  & X & Y \\
                           X^\intercal  & B & Z \\
                           Y^\intercal  & Z^\intercal  & C \rmatrix
\lmatrix v \\ w \\ -C^{-1}Y^\intercal  v \rmatrix \\[0.8ex]
    &=    2 v^\intercal  (X-YC^{-1}Z^\intercal ) w + w^\intercal  B w .
\end{split}
\end{equation*}
This can only be true for all $v$ and $w$ if $X=YC^{-1}Z^\intercal $.
Moreover, this latter condition (together with the positivity of $V_{ABC}$)
is enough to guarantee that~\eqref{eq cond eq1} is satisfied. Indeed, we can write
\begin{align*}
  \lmatrix YC^{-1}Y^\intercal  & YC^{-1}Z^\intercal  & Y \\
                  ZC^{-1}Y^\intercal  & B & Z \\
                  Y^\intercal  & Z^\intercal  & C \rmatrix &= \lmatrix 0 & & \\ & B-ZC^{-1}Z^\intercal  & \\ & & 0 \rmatrix  \\[0.8ex]
    &\quad + \lmatrix YC^{-\frac12} \\ ZC^{-\frac12} \\ C^{\frac12} \rmatrix \lmatrix YC^{-\frac12} \\ ZC^{-\frac12} \\ C^{\frac12} \rmatrix^\intercal \\[0.8ex]
    & \geq 0 ,
\end{align*}
where $B-ZC^{-1}Z^\intercal \geq 0$ follows from
$\left(\begin{smallmatrix} B & Z \\ Z^\intercal & C \end{smallmatrix}\right) \geq 0$.

\item[$4\!\Rightarrow\! 5$.] If in~\eqref{N} we define
\begin{equation}\begin{split}
  H &= H_R \coloneqq \begin{pmatrix} \mathds{1} & 0 \\ 0 & ZC^{-1} \\ 0 & \mathds{1} \end{pmatrix} \text{ and}\\[0.8ex]
  K &= K_R \coloneqq \begin{pmatrix} 0 & & \\ & B-ZC^{-1}Z^\intercal  & \\ & & 0 \end{pmatrix} ,
  \label{gaus Petz eq2}
\end{split}\end{equation}
we obtain straightforwardly
\begin{align*}
  (I_A\oplus N_{C\rightarrow BC}) (V_{AC}) &= H_R \lmatrix A & X \\ X^\intercal  & C \rmatrix H_R^\intercal + K_R  \\[0.8ex]
    &= \lmatrix A & X & Y \\ X^\intercal  & B & Z \\ Y^\intercal  & Z^\intercal  & C \rmatrix  \\[0.8ex]
    &= V_{ABC} ,
\end{align*}
provided that $X=YC^{-1}Z^{\intercal}$. We will see in the next section
that this map is nothing but a specialisation to the Gaussian case of a
general construction known as transpose channel, or Petz recovery map.

\item[$5\!\Rightarrow\! 2$.] Since it is known that classical Gaussian channels acting on
$C$ always increase $V_{AC}/V_C$~\cite{Lami16}, it is clear that the equality in~\eqref{INEQ 1}
is a necessary condition for the existence of a Gaussian recovery map $N_{C\rightarrow BC}$.
\end{description}
\end{proof}

\medskip

\section{Gaussian recoverability}
\label{sec:Gaussian-recov}
Here, we discuss the role of some well-known remainder terms for inequalities
of the form~\eqref{SSA}. These terms have been introduced recently in the
context of sufficient statistics~\cite{Petz-old} and its approximate
variants~\cite{VV}, or so-called ``recoverability''.
In~\cite{FawziRenner}, a form involving recovery maps was proposed for such
a term in the fully quantum case (i.e., considering the SSA for von Neumann entropy)
based on the fidelity of recovery, and subsequently strengthened to a bound involving the
measured relative entropy~\cite{BHOS}; in both cases the given bounds
turn out to be operationally meaningful quantities~\cite{CHMOSWW}.
The much simpler classical reasoning (with a better bound) was presented in~\cite{VV}.
We will translate these results into the Gaussian setting
in order to find an explicit expression for a remainder term to be added to~\eqref{SSA}.

For classical probability distributions $p$ and $q$ over a discrete alphabet,
in~\cite{VV} the following inequality was shown, which improves on the
monotonicity of the relative entropy under channels:
\begin{equation}
  D(p\|q) - D(Np\|Nq) \geq D\left( p \| R N p \right) ,
  \label{recov}
\end{equation}
where $N = (N_{ji})$ is any stochastic map (channel)
and the action of the transpose channel
(also known as \emph{Petz recovery map}~\cite{Petz book,BarnumKnill})
$R=R_{q,N}$ on an input distribution $r$ is uniquely
defined via the requirement that $N_{ji}q_i = R_{ij} (Nq)_j$ for all $i$ and $j$.
Explicitly,
\begin{equation}
  (R_{q,N}\, r)_i \coloneqq \sum_j \frac{q_i N_{ji}}{(Nq)_j} r_j .
  \label{Petz}
\end{equation}
Observe that $R_{q,N}$ is a bona fide channel, since
\begin{equation*}
  \sum_i (R_{q,N})_{ij} = \sum_i \frac{q_i N_{ji}}{(Nq)_j}
                        = \frac{(Nq)_j}{(Nq)_j}
                        = 1 .
\end{equation*}
For obvious reasons, we will call the right hand side of~\eqref{recov} the
{\it relative entropy of recovery}. The proof of~\eqref{recov} is a simple
application of the concavity of the logarithm, and we
reproduce it here for the benefit of the reader.
\begin{align}
  D\left(p\|R_{q,N}Np\right)
     &= \sum_i p_i \Big(\ln p_i - \ln (R_{q,N}Np)_i \Big) \nonumber\\[0.8ex]
     &= \sum_i p_i \Big(\ln p_i - \ln \sum_j \frac{q_i\,N_{ji}}{(Nq)_j}\,(Np)_j \Big) \label{petz eq1} \\
     &\leq \sum_i p_i \Big(\ln p_i - \sum_j N_{ji} \ln \frac{q_i}{(Nq)_j}\,(Np)_j \Big) \label{petz eq2} \\
     &= D(p\|q) - D\left(Np\|Nq\right) . \nonumber
  \label{petz eq2}
\end{align}

Although we wrote out the proof only for random variables taking values in a discrete
alphabet, all of the above expressions make perfect sense also in more general cases,
e.g. when $i$ and $j$ are multivariate real variables.
If $N$ is a classical Gaussian channel acting as in~\eqref{N}, it
can easily be verified that the `transition probabilities' $N(x,y)$ satisfying
\begin{equation}
  (Np)(x) = \int dy\, N(x,y) p(y)
\end{equation}
take the form
\begin{equation}
  N(x,y) = \frac{e^{-\frac{1}{2} (x-Hy)^\intercal K^{-1} (x-Hy)}}{\sqrt{(2\pi)^n\det K}} .
  \label{N Gauss}
\end{equation}

Following again~\cite{VV}, we observe that if the output of the random channel
$N$ is a deterministic function of the input, then~\eqref{recov} is always
saturated with equality.
This can be seen by noticing that in that case for all $i$ there is only one
index $j$ such that $N_{ji}\neq 0$ (and so $N_{ji}=1$). Therefore, the step
from~\eqref{petz eq1} to~\eqref{petz eq2} is an equality. There is a very special
case when this remark is useful. Consider a triple of random variables $XYZ$
distributed according to $p(xyz)$, a second probability distribution $q(xyz)=p(x)p(yz)$,
and the channel $N$ consisting of discarding $Y$. Obviously, in this case the
output is a deterministic function of the input. It is easily seen that the
reconstructed global probability distribution $R_{q,N}Np$ is
\begin{equation}
  \tilde{p}(xyz) = p(xz) p(y|z) .
\end{equation}
Then the saturation of~\eqref{recov} allows us to write
\begin{equation}
  I(X:Y|Z) = D(p\|q) - D(Np\|Nq) = D(p\|\tilde{p}) .
  \label{I rel ent}
\end{equation}

\subsection{Gaussian Petz recovery map}
\label{gaus Petz}
From now on, we will consider the case in which $N$ is a {classical Gaussian channel} transforming covariance matrices according to the rule~\eqref{N}. As can be easily verified, if also $q$ is a multivariate Gaussian distribution, then $R_{q,N}$ becomes a {classical Gaussian channel} as well. We compute its action in the case we are mainly interested in, that is, when the left--hand side of~\eqref{recov} corresponds to the difference of the two sides of~\eqref{SSA}, and verify that it coincides with the recovery map introduced in Section~\ref{sec exact recov} (via the general action~\eqref{N} with the substitutions~\eqref{gaus Petz eq2}).

\begin{prop}
  \label{prop Petz}
  Let $q$ be a tripartite Gaussian probability density with zero mean and covariance matrix
  \begin{equation*}
    V_A\oplus V_{BC} = \begin{pmatrix} A & 0 & 0 \\ 0 & B & Z \\ 0 & Z^\intercal & C \end{pmatrix} ,
  \end{equation*}
  and let the channel $N$ correspond to the action of discarding the $B$ components, i.e.
  $H = \Pi_{AC}
   = \left( \begin{smallmatrix} \mathds{1} & 0 & 0 \\ 0 & 0 & \mathds{1} \end{smallmatrix}\right)$
  and $K=0$ in~\eqref{N}. Then, the action $C\rightarrow BC$ of the Petz recovery
  map~\eqref{Petz} on Gaussian variables with zero mean can be written at the level
  of covariance matrices as~\eqref{N}, where $H_R$ and $K_R$ are given by~\eqref{gaus Petz eq2}.
\end{prop}

\begin{proof}
The Petz recovery map~\eqref{Petz} is a composition of three operations: first the pointwise division by a Gaussian distribution, then the transpose of a deterministic channel, and eventually another pointwise Gaussian multiplication. It should be obvious from~\eqref{gauss} that a pointwise multiplication by a Gaussian distribution with covariance matrix $A$ is a Gaussian (non--deterministic) channel that leaves the mean vector invariant and acts on covariance matrices as $V\mapsto V'=(V^{-1}+A^{-1})^{-1}$. Furthermore, it can be proven that the transpose $N^\intercal$ of the channel $N$ in~\eqref{N} sends Gaussian variables with zero mean to other Gaussian variables with zero mean, while on the inverses of the covariance matrices it acts as
\begin{equation}
  N^\intercal : V^{-1} \longmapsto\ (V')^{-1} = H^\intercal (V+K)^{-1} H .
  \label{N^T}
\end{equation}
{A way to prove the above equation is by using~\eqref{N Gauss} to compute directly
the action of $N^\intercal$ on a Gaussian input distribution.}

After the preceding discussion, it should be clear that under our hypotheses
the action of the Petz recovery map can be written as
\begin{equation}
\begin{split}
  \sigma_{AC} \longmapsto
    \sigma'_{ABC} &= \Big( V_A^{-1}\oplus V_{BC}^{-1} \\[0.8ex]
                          &\quad + (\sigma_{AC}^{-1}-V_A^{-1}\oplus V_C^{-1})\oplus 0_B \Big)^{-1} .
\end{split}
  \label{gaus Petz eq3}
\end{equation}
The Woodbury matrix identity (see~\cite{Woodbury}, or~\cite[Equation (6.0.10)]{SCbook}),
\begin{equation}
  (S+UTV)^{-1} = S^{-1} - S^{-1}U\left(VS^{-1}U+T^{-1}\right)^{-1} V S^{-1},
  \label{Wood}
\end{equation}
can be used to bring~\eqref{gaus Petz eq3} into the canonical form~\eqref{N}:
\begin{align*}
  \sigma'_{ABC} &= \left( V_A^{-1}\oplus V_{BC}^{-1}
                   + (\sigma_{AC}^{-1}-V_A^{-1}\oplus V_C^{-1})\oplus 0_B \right)^{-1} \\[0.8ex]
                &= \big( V_A^{-1}\oplus V_{BC}^{-1} \\
                   &\quad + \Pi_{AC}^\intercal  (\sigma_{AC}^{-1}-V_A^{-1}\oplus V_C^{-1}) \Pi_{AC} \big)^{-1} \\[0.8ex]
                &= V_A\oplus V_{BC} \\
                &\quad- (V_A\oplus V_{BC}) \Pi_{AC}^\intercal \\
                     &\quad\quad \cdot \Big( (\sigma_{AC}^{-1}\! - V_A^{-1}\oplus V_C^{-1})^{-1} \\
                     &\quad\quad\quad + \Pi_{AC} (V_A \oplus V_{BC}) \Pi_{AC}^\intercal\Big)^{-1} \\
                     &\quad\quad \cdot \Pi_{AC} (V_A \oplus V_{BC}) \\[0.8ex]
                &= V_A \oplus V_{BC} - (V_A \oplus V_{BC}) \Pi_{AC}^\intercal \\
                &\quad \cdot \Big( -V_A \oplus V_C \\
                &\quad\quad - (V_A \oplus V_C) (\sigma_{AC}-V_A \oplus V_C)^{-1} (V_A \oplus  V_C) \\
                &\quad\quad+ V_A \oplus V_C \Big)^{-1} \\
                &\quad \cdot \Pi_{AC} (V_A \oplus V_{BC}) \\[0.8ex]
                &= V_A\oplus V_{BC} \\
                &\quad+ (V_A\oplus V_{BC}) \Pi_{AC}^\intercal (V_A^{-1} \oplus V_C^{-1}) \\
                &\quad\quad \cdot(\sigma_{AC}-V_A\oplus V_C) \\
                &\quad\quad\cdot (V_A^{-1}\oplus V_C^{-1}) \Pi_{AC} (V_A\oplus V_{BC}) \\[0.8ex]
                &= H_R \sigma_{AC} H_R^\intercal + K_R\, ,
\end{align*}
where we have employed the definitions
\begin{align*}
  H_R &= (V_A\oplus V_{BC}) \Pi_{AC}^\intercal (V_A^{-1}\oplus V_C^{-1})
       = \lmatrix \mathds{1} & 0 \\ 0 & ZC^{-1} \\ 0 & \mathds{1} \rmatrix \text{ and} \\[0.8ex]
  K_R &= \lmatrix 0 & & \\ & B-ZC^{-1} Z^\intercal & \\ & & 0 \rmatrix .
\end{align*}
\end{proof}

\subsection{Gaussian relative entropy of recovery}
We are ready to employ the classical theory of recoverability in order
to find the expression of the relative entropy of recovery in the Gaussian case.

\begin{prop}
  For all tripartite covariance matrices $V_{ABC}>0$ written in block
  form as in~\eqref{global CM}, we have
  \begin{equation}
  \begin{split}
    I_{M}(A:B|C)_{V} &= \frac12 \ln\frac{\det V_{AC}\det V_{BC}}{\det V_{ABC}\det V_C} \\
                     &= D\!\left( V_{ABC} \| \tilde{V}_{ABC}\right) ,
  \end{split}
  \label{SSA+}
  \end{equation}  where
  \begin{equation}
    \tilde{V}_{ABC} \coloneqq \begin{pmatrix} A & YC^{-1} Z^\intercal  & Y \\
                                           ZC^{-1}Y^\intercal  & B & Z \\
                                           Y^\intercal  & Z^\intercal  & C \end{pmatrix}
    \label{V tilde}
  \end{equation}
  and the relative entropy function $D(\cdot\|\cdot)$ is given by~\eqref{rel ent}.
\end{prop}

\begin{proof}
This is just an instance of~\eqref{I rel ent} applied to the continuous
Gaussian variable $(X_{A},X_{B},X_{C})$.
\end{proof}

The identity~\eqref{SSA+} is useful in deducing new constraints that will
be much less obvious coming from a purely matrix analysis perspective.
For instance, it is well known that $D(p\|q)\geq -\ln \mathcal{F}^2(p,q)$ {(see e.g.~\cite{newRenyi,KMRA})},
where the fidelity is given by $\mathcal{F}(p,q)=\sum_i \sqrt{p_i q_i}$
in the discrete case.
In case of Gaussian variables with the same mean, it holds
\begin{equation}
  \mathcal{F}^2(p_A,p_B) = \frac{\det (A!B)}{\sqrt{\det A\det B}} ,
  \label{fid gaus}
\end{equation}
where $(A!B)\coloneqq 2\left(A^{-1}+B^{-1}\right)^{-1}$ is the \emph{harmonic mean}
of $A$ and $B$.
Inserting this standard lower bound into~\eqref{SSA+} we obtain
\begin{equation}
  \frac{\det V_{AC}\det V_{BC}}{\det V_{ABC}\det V_C}
    \geq \frac{\det V_{ABC}\det \tilde{V}_{ABC}}{\left( \det (V_{ABC}!\tilde{V}_{ABC}) \right)^2},
  \label{SSA+ fid}
\end{equation}
leading to
\begin{equation}
   I_{M}(A:B|C)_{V}
    \geq \frac12 \ln\frac{\det V_{ABC}\det \tilde{V}_{ABC}}{\left( \det (V_{ABC}!\tilde{V}_{ABC}) \right)^2}.
  \label{SSA+ fid2}
\end{equation}
Using furthermore
\begin{align*}
  \det \tilde{V}_{ABC} &= \det\tilde{V}_{BC} \det (\tilde{V}_{ABC}/\tilde{V}_{BC}) \\
                       &= \det V_{BC} \det (\tilde{V}_{AC}/\tilde{V}_{C}) \\
                       &= \det V_{BC} \det (V_{AC}/V_{C}) ,
\end{align*}
we also arrive at the inequality
\begin{equation}
  \det V_{ABC} \leq \det (V_{ABC}!\tilde{V}_{ABC}) .
\end{equation}
To illustrate the power of this relation,
we note that inserting the harmonic-geometric mean inequality
for matrices~\cite[Corollary 2.1]{Ando79}
\begin{equation*}
  A!B \leq A\# B
\end{equation*}
yields again SSA~\eqref{SSA} in the form
$\det \tilde{V}_{ABC} \geq \det V_{ABC}$.

\section{A lower bound on $I_{M}(A:B|C)_{V}$}
\label{sec:I_M-lower-bound}
Throughout this section, we explore some ways of strengthening Theorem~\ref{thm satur}
by finding a suitable lower bound on the log-det conditional mutual information $I_{M}(A:B|C)_{V}$. The expression we are
seeking should have two main features: (a) it should be easily computable in
terms of the blocks of $V_{ABC}$; and (b) the explicit saturation condition in Theorem~\ref{thm satur}(4) should be easily readable from it. This latter requirement can be accommodated, for example, if the lower bound involves some kind of distance between the off-diagonal block $X$ and its `saturation value' $YC^{-1}Z^{\intercal}$. We start with a preliminary result.

\begin{prop}
  \label{prop I(A:B)}
  For all matrices
  \begin{equation*}
    V_{AB} = \begin{pmatrix} A & X \\ X^{\intercal} & B \end{pmatrix} \geq 0,
  \end{equation*}
  we have
  \begin{equation}
  \begin{split}
    I_{M}(A:B)_{V} &\geq \frac12 \Tr [A^{-1} X B^{-1} X^{\intercal}] \\
                   &= \frac12 \big\| A^{-1/2} X B^{-1/2} \big\|^{2}_{2} .
  \end{split}
  \label{prop I(A:B) eq}
  \end{equation}
\end{prop}

\begin{proof}
Using, in this order, the standard factorisation of the determinant in terms
of the Schur complement, the identity $\ln \det V = \Tr \ln V$ (where $V>0$),
and the inequality $\ln(\mathds{1}+\Delta)\leq \Delta$ (for Hermitian
$\Delta > -\mathds{1}$), we find
\begin{align*}
  I_{M}(A:B)_{V} &=    \frac12 \ln \frac{\det V_{A}\det V_{B}}{\det V_{AB}} \\[0.8ex]
                 &=   -\frac12 \ln \det V_{A}^{-1/2}(V_{AB}/V_{B})V_{A}^{-1/2} \\[0.8ex]
                 &=   -\frac12 \ln \det (\mathds{1} - A^{-1/2}XB^{-1}X^\intercal A^{-1/2}) \\[0.8ex]
                 &=   -\frac12 \Tr \ln (\mathds{1} - A^{-1/2}XB^{-1}X^\intercal A^{-1/2})  \\[0.8ex]
                 &\geq \frac12 \Tr A^{-1/2}XB^{-1}X^\intercal A^{-1/2} \\[0.8ex]
                 &=    \frac12 \Tr A^{-1}XB^{-1}X^\intercal \\[0.8ex]
                 &=    \frac12  \big\| A^{-1/2} X B^{-1/2} \big\|^{2}_{2}\, ,
\end{align*}
and we are done.
\end{proof}

\begin{thm}
  \label{thm lower b}
  For all $V_{ABC}>0$ written in block form as in~\eqref{global CM}, we have
  the following chain of inequalities:
  \begin{align}
    I_{M}(A:B|C)_{V}
      &\geq \frac12 \Tr\Big[ (V_{AC}/V_{C})^{-1} (X-YC^{-1}Z^{\intercal}) \nonumber\\
      &\quad \cdot (V_{BC}/V_{C})^{-1} (X-YC^{-1}Z^{\intercal})^{\intercal} \Big] \\[0.8ex]
      &\geq \frac12 \Tr\Big[ A^{-1} (X-YC^{-1}Z^{\intercal}) \nonumber \\
      &\quad \cdot B^{-1} (X-YC^{-1}Z^{\intercal})^{\intercal} \Big] \\[0.8ex]
      &= \frac12 \left\| A^{-1/2}(X-YC^{-1}Z^{\intercal})B^{-1/2} \right\|^{2}_{2}\, .
    \label{lower b}
  \end{align}
\end{thm}

\begin{proof}
We want to use the identity~\eqref{I cond inv} to lower bound $I_{M}(A:B|C)_{V}$.
In order to do so, we need to write out the $A$-$B$ off-diagonal block of the
inverse $(V_{ABC})^{-1}$. With the help of the projectors onto the $A$ and $B$
components, denoted by $\Pi_{A}$ and $\Pi_{B}$ respectively, we are seeking an
explicit expression for $\Pi_{A} (V_{ABC})^{-1} \Pi_{B}^{\intercal}$. Remember
that the block-inversion formula~\eqref{inv} gives
\begin{align}
  \Pi_{1} (W_{12})^{-1} \Pi_{1}^{\intercal} &= (W_{12}/W_{2})^{-1} , \label{inv1} \\[0.8ex]
  \Pi_{1} (W_{12})^{-1} \Pi_{2}^{\intercal} &= - W_{1}^{-1} (\Pi_{1}W_{12}\Pi_{2}^{\intercal})
                                                            (W_{12}/W_{1})^{-1} , \label{inv2}
\end{align}
for an arbitrary bipartite block matrix $W_{12}$. This allows us to write
\begin{align*}
  \Pi_{A} (V_{ABC})^{-1} \Pi_{B}^{\intercal}
     &= \Pi_{A} \Pi_{AB} (V_{ABC})^{-1} \Pi_{AB}^{\intercal} \Pi_{B}^{\intercal} \\[0.8ex]
     &= \Pi_{A} (V_{ABC}/V_{C})^{-1} \Pi_{B}^{\intercal} \\[0.8ex]
     &= -(V_{AC}/V_{C})^{-1} \bigl( \Pi_{A} V_{ABC}/V_{C} \Pi_{B}^{\intercal} \bigr) \\
     &\quad \cdot \bigl( (V_{ABC}/V_{C}) \big/ (V_{AC}/V_{C}) \bigr)^{-1} \\[0.8ex]
     &= -(V_{AC}/V_{C})^{-1} \bigl( X - YC^{-1} Z^{\intercal} \bigr) \\
     &\quad \cdot (V_{ABC}/V_{AC})^{-1} .
\end{align*}
Exchanging $A$ and $B$ in this latter expression and taking
subsequently the transpose we arrive also at
\begin{equation*}
\begin{split}
  \Pi_{A} (V_{ABC})^{-1} \Pi_{B}^{\intercal}
     &= -(V_{ABC}/V_{BC})^{-1} \big( X - YC^{-1} Z^{\intercal} \big) \\
     &\quad \cdot (V_{BC}/V_{C})^{-1} .
\end{split}
\end{equation*}
Now we are ready to invoke Proposition~\ref{prop I(A:B)} to write
\begin{align*}
  & I_{M}(A:B|C)_{V} \\[0.8ex]
  &\quad = I_{M}(A:B)_{V^{-1}} \\[0.8ex]
           &\quad\geq \frac12 \Tr \Big[(V^{-1})_{A}^{-1} (\Pi_{A} V^{-1} \Pi_{B}^{\intercal}) \\
           &\quad\quad \cdot (V^{-1})_{B}^{-1} (\Pi_{B}^{\intercal} V^{-1} \Pi_{A}) \Big] \\[0.8ex]
           &\quad=    \frac12 \Tr \Bigl[ (V_{ABC}/V_{BC}) \\
           &\quad\quad \cdot \bigl( (V_{ABC}/V_{BC})^{-1} (X-YC^{-1} Z^{\intercal}) (V_{BC}/V_{C})^{-1} \bigr) \\
           &\quad\quad \cdot (V_{ABC}/V_{AC}) \\
           &\quad\quad \cdot \bigl( (V_{AC}/V_{C})^{-1} (X-YC^{-1} Z^{\intercal}) (V_{ABC}/V_{AC})^{-1} \big)^{\intercal} \Bigr] \\[0.8ex]
           &\quad= \frac12 \Tr \Big[ (V_{AC}/V_{C})^{-1} ( X - YC^{-1} Z^{\intercal}) \\
           &\quad\quad \cdot (V_{BC}/V_{C})^{-1} ( X - YC^{-1} Z^{\intercal})^{\intercal} \Big]\, .
\end{align*}
Since on the one hand $V_{AC}/V_{C}\leq V_{A}=A$, and on the other hand the
expression $\Tr R K S K^{\intercal}$ is clearly monotonic in $R,S\geq 0$,
we finally obtain
\begin{align*}
  I_{M}(A:B|C)_{V} &\geq \frac12 \Tr \bigl[ A^{-1} ( X - YC^{-1} Z^{\intercal}) \\
  &\quad \cdot B^{-1} ( X - YC^{-1} Z^{\intercal})^{\intercal} \bigr] \\[0.8ex]
                   &= \frac12 \bigl\| A^{-1/2}(X-YC^{-1}Z^{\intercal})B^{-1/2} \bigr\|^{2}_{2}\, .
\end{align*}
\end{proof}

It can easily be seen that the above result satisfies the requirements stated in the beginning of the section, i.e. it is easily computable in terms of the blocks of $V_{ABC}$ and it is faithful.

We are now ready to start the investigation of {\em quantum} covariance matrices in the next section.

\section{Strengthenings of SSA for quantum covariance matrices \protect\\ and \mbox{R\'enyi-$2$} Gaussian squashed entanglement}
\label{sec:Renyi-2-Gaussian-squashed}

\subsection{Gaussian states in quantum optics}
In this final section we show how to apply results on log-det conditional
mutual information to infer properties of Gaussian states in quantum optics.
Before doing so, let us provide a very brief introduction to quantum optics,
a framework of great importance for practical applications and
implementations of quantum communication protocols.
The set of $n$ electromagnetic modes that are available for transmission
of information translates to a set of $n$ pairs of canonical operators
$x_i, p_j$ ($i=1,\ldots, n$) acting on an infinite-dimensional Hilbert space
and obeying the canonical commutation relations
$[x_i, p_j] = i\delta_{ij}$ (in natural units with $\hbar=1$). These operators are the non-commutative
analogues of the classical electric and magnetic fields.
By introducing the vector notation $r\coloneqq (x_1,p_1,\ldots,x_n,p_n)^{\intercal}$
we can rewrite the canonical commutation relations in the more convenient form
\begin{equation}
[r,r^\intercal ] = i \Omega \coloneqq i {\begin{pmatrix} 0 & 1 \\ -1 & 0 \end{pmatrix}\!}^{\oplus n} = i\begin{pmatrix} 0 & \mathds{1} \\ -\mathds{1} & 0 \end{pmatrix} ,
\label{CCR}
\end{equation}
where $\Omega$ is called the \emph{standard symplectic form}. The antisymmetric, non-degenerate quadratic form identified by $\Omega$ is called \emph{standard symplectic product}, and the linear space $\mathds{R}^{2n}$ endowed with this product is a \emph{symplectic space}. In what follows, the symplectic space associated with a quantum optical system $A$ will be denoted with $\Sigma_A$. For an introduction to symplectic geometry, we refer the reader to the excellent monograph~\cite{deGosson}.

Following the formalism of quantum mechanics, we represent states
as \emph{density matrices}, i.e. positive semidefinite, trace class
operators acting on the background Hilbert space. For the probabilistic
interpretation of measurements to be consistent, we assume any density
matrix $\rho$ to have unit trace, i.e. $\Tr \rho=1$. Exactly as in the classical
case, also for quantum electromagnetic fields the Hamiltonian is quadratic
in the canonical operators. Thus, not surprisingly, the states that are most
frequently produced in the laboratories are thermal states of quadratic
Hamiltonians of the form $\mathcal{H}=\frac12 r^\intercal H r$, where $H>0$
is a $2n\times 2n$ real, positive definite matrix. These states are so special
that they deserve a name on their own, being called {\em Gaussian states}
\cite{biblioparis,weedbrook12,Adesso14}. The reason is intuitively clear:
since a thermal state of a system with Hamiltonian $\mathcal{H}$ is well-known to
be representable as $\rho = \frac{e^{-\beta \mathcal{H}}}{\mathcal{Z}}$,
where $\mathcal{Z}$ is a normalisation constant and $\beta=1/kT$ is the
inverse temperature, it is clear how a quadratic Hamiltonian produces an
expression resembling a Gaussian function.\footnote{This intuitive reason is in
fact supported by more substantial arguments. Namely, Gaussian states
are also identified by a Gaussian Wigner function, as written in~\eqref{Wigner}.}

For a quantum state described by a density matrix $\rho$ the
first moments are given by the expected value of the field operators,
in turn expressible as $s=\Tr [\rho r]$.
However, as expected, the information-theoretical properties of
Gaussian states can be fully understood in terms of the second-moment
correlations they display, encoded in the $2n\times 2n$ covariance matrix $V$
whose entries are
\begin{equation}
   V_{ij} \coloneqq \Tr \left[ \rho \left\{(r-s)_i, (r-s)_j \right\} \right] ,
   \label{QCM}
\end{equation}
where the anticommutator $\{H,K\}\coloneqq HK+KH$ is needed
in the quantum case in order to make the above expression real, and $s\coloneqq s\cdot \text{id}$ as operators on the Hilbert space.\footnote{It is customary
not to divide by $2$ when defining the covariance matrix in the quantum case.
The reason will become apparent in a moment.}
Any quantum state $\rho$ of an $n$-mode electromagnetic field can be equivalently described in terms of phase space quasi-probability distributions, such as the Wigner distribution~\cite{Wigner}. Hence Gaussian states can be defined, in general, as the continuous variable states with a Gaussian Wigner distribution, given by
\begin{equation}\label{Wigner}
W_\rho(\xi) \coloneqq \frac{1}{\pi^n \sqrt{\det V}} e^{-(\xi-s)^\intercal V^{-1} (\xi-s)},
\end{equation}
in terms of the vector of first moments $s$ and the QCM $V$, with $\xi \in \mathds{R}^{2n}$ a phase space coordinate vector.

Let us have a closer look at the set of matrices arising from~\eqref{QCM}. Differently from
what happens in the classical case, not every positive definite matrix $V>0$
can be the covariance matrix of a Gaussian state. In fact, Heisenberg's uncertainty
principle imposes further constraints, quantum mechanical in nature. It turns out
\cite{simon94} that covariance matrices of quantum states (not necessarily Gaussian)
must obey the inequality
\begin{equation}
V\geq i\Omega . \label{Heisenberg}
\end{equation}
Furthermore, all $2n\times 2n$ real matrices satisfying~\eqref{Heisenberg},
collectively called \emph{quantum covariance matrices} (QCMs) can
be covariance matrices of suitably chosen Gaussian states.
Therefore, according to our convenience, we can think of Gaussian states as operators on the background Hilbert space, or we can adopt the complementary picture at the symplectic space level, and parametrise Gaussian states with their covariance matrices.

Clearly, linear transformations $r \rightarrow S r$ that preserve the commutation relations~\eqref{CCR} play a special role within this framework. Any such transformation is described by a {\it symplectic} matrix, i.e.~a matrix $S$ with the property that $S\Omega S^\intercal=\Omega$. Symplectic matrices form a non-compact, connected Lie group that is additionally closed under transposition, and is typically denoted by $\mathrm{Sp}(2n,\mathds{R})$~\cite{pramana}. The importance of these operations arises from the fact that for any symplectic $S$ there is a unitary evolution $U_S$ on the Hilbert space such that $U_S^\dag r U_S = Sr$. When a unitary conjugation $\rho\mapsto U_S \rho U_S^\dag$ is applied to a state $\rho$, its covariance matrix transforms as $V\mapsto SVS^\intercal$. Accordingly, observe that~\eqref{Heisenberg} is preserved under congruences by symplectic matrices. It turns out that under such congruences positive matrices can be brought into a remarkably simple form.

\begin{lemma}[Williamson's decomposition~\cite{willy,willysim}] \label{lemma Williamson}
Let $K>0$ be a positive, $2n\times 2n$ matrix. Then there is a symplectic transformation $S$ such that $K= S \Delta S^\intercal$, where according to the block decomposition~\eqref{CCR} one has $\Delta= \lmatrix D & 0 \\ 0 & D \rmatrix$, and $D$ is a positive diagonal matrix whose nonzero entries depend (up to their order) only on $K$, and are called symplectic eigenvalues.
\end{lemma}

Thanks to Williamson's decomposition, we see that~\eqref{Heisenberg} can be cast into the simple form $D\geq \mathds{1}$, and that the minimal elements in the set of QCMs are exactly those matrices $V$ for which one of the following equivalent conditions is met: (a) $D=\mathds{1}$; (b) $\det V=1$; (c) $\rk (V \pm i\Omega) = n$ (i.e. half the maximum). These special QCMs are called ``pure'', since the corresponding Gaussian state is a rank-one projector.

When the system under examination is made of several parties (each comprising a certain number
of modes), the global QCM will have a block structure as in~\eqref{global CM}.
The symplectic form in this case is simply given by the direct sum of
the local symplectic forms, e.g. for a composite system $AB$ one has $\Omega_{AB}=\Omega_{A}\oplus \Omega_{B}$. This can be rephrased by saying that the symplectic space associated with the system $AB$ is the direct sum of the symplectic spaces associated with $A$ and $B$, in formula $\Sigma_{AB}=\Sigma_{A}\oplus \Sigma_{B}$~\cite[Equation (1.4)]{deGosson}.
Conversely, discarding a subsystem corresponds to performing an orthogonal projection of the QCM onto the corresponding symplectic subspace~\cite[Section 1.2.1]{deGosson}, in formula $V_{A}=\Pi_{A} V_{AB} \Pi_{A}^\intercal$.

Pure Gaussian states enjoy many useful properties that we will exploit multiple times throughout this section. To explore them, a clever use of the complementarity between the two pictures at the Hilbert space level and at the QCM level is of prime importance. Let us illustrate this point by presenting three lemmas we will make use of in deriving the main results of this section.

\begin{lemma} \label{lemma pure reduction}
Let $V_{AB}$ be a QCM of bipartite system $AB$. Denote by $V_{A}=\Pi_{A} V_{AB} \Pi_{A}^\intercal$ the reduced QCM corresponding to the subsystem $A$, and analogously for $V_{B}$. If $V_{A}$ is pure, then $V_{AB} = V_{A} \oplus V_{B}$.
\end{lemma}

\begin{proof}
The statement becomes obvious at the Hilbert space level. In fact, the reduced state on $A$ of a bipartite state $\rho_{A B}$ is given by $\rho_{A} = \Tr_{B} \rho_{A B}$, where $\Tr_{B}$ denotes partial trace~\cite{Nielsen}. Evaluating the ranks of both sides of this equation shows that if $\rho_{A}$ is pure then the global state must be factorised.
\end{proof}

Extending the system as to include auxiliary degrees of freedom is a standard technique in quantum information, popularly referred to as going to the ``Church of the larger Hilbert space''~\cite{Church}. Such a technique can be most notably employed in order to \emph{purify} the system under examination, as detailed in the following lemma~\cite{G purif}.

\begin{lemma} \label{lemma pur}
For all QCMs $V_A$ pertaining to a system $A$ there exists an extension $AE$ of $A$ and a pure QCM $\gamma_{AE}$ such that $\Pi_A \gamma_{AE} \Pi_A^\intercal = V_A$, where $\Pi_A$ is the projector onto the symplectic subspace $\Sigma_A\subset \Sigma_{AE}$.
\end{lemma}

\begin{proof}
See~\cite[Section III.D]{G purif}.
\end{proof}

Let us present here another useful observation.

\begin{lemma} \label{lemma fact out}
For all QCMs $V_A\geq i\Omega_A$ of a system $A$, there is a decomposition $\Sigma_A=\Sigma_{A_1} \oplus \Sigma_{A_2}$ of the global symplectic space into a direct sum of two symplectic subspaces such that
\begin{equation}
V_A = V_{A_1} \oplus \eta_{A_2} ,
\end{equation}
where $V_{A_1} > i\Omega_{A_1}$ and $\eta_{A_2}$ is a pure QCM. Furthermore, for every purification $\gamma_{AE}$ of $V_A$ (see Lemma~\ref{lemma pur}) there is a symplectic decomposition of $E$ as $\Sigma_{E}=\Sigma_{E_1} \oplus \Sigma_{E_2}$ such that: (a) $\gamma_{AE}=\gamma_{A_1E_1} \oplus \eta_{A_2} \oplus \tau_{E_2}$, with $\eta_{A_2}, \tau_{E_2}$ pure QCMs; (b) $n_{A_1} = n_{E_1}$; and (c) $\gamma_{E_1}>i\Omega_{E_1}$.
\end{lemma}

\begin{proof}
The first claim is a direct consequence of Williamson's decomposition, Lemma~\ref{lemma Williamson}. The subspace $\Sigma_{A_2}$ corresponds to those symplectic eigenvalues of $V_A$ that are equal to $1$. 

Now, let us prove the second claim. Consider an arbitrary pure QCM $\gamma_{AE}$ that satisfies $\gamma_A=V_A=V_{A_1}\oplus \eta_{A_2}$. Since in particular $\gamma_{A_2}=\eta_{A_2}$, we can apply Lemma~\ref{lemma pure reduction} and conclude that $\gamma_{AE}=\gamma_{A_1 E} \oplus \eta_{A_2}$. The first claim of the present lemma tells us that $\gamma_{E}=\gamma_{E_1}\oplus \tau_{E_2}$, with $\gamma_{E_1}> i\Omega_{E_1}$ and $\tau_{E_2}$ pure. Again, Lemma~\ref{lemma pure reduction} yields $\gamma_{AE}=\gamma_{A_1E_1}\oplus \eta_{A_2}\oplus \tau_{E_2}$, corresponding to statement (b). Hence, we have only to show that $n_{A_1}=n_{E_1}$. In order to show this, let us write
\begin{equation*}
\gamma_{A_1 E_1} = \begin{pmatrix} V_{A_1} & L \\ L^\intercal & \gamma_{E_1} \end{pmatrix} .
\end{equation*}
We can invoke~\cite[Equation (8)]{Lami16} to deduce the identity $V_{A_1} - L\gamma_{E_1}^{-1} L^\intercal = \Omega V_{A_1}^{-1} \Omega^\intercal$, that is, $L \gamma_{E_1}^{-1} L^\intercal = V_{A_1} - \Omega V_{A_1}^{-1} \Omega^\intercal$. Since the right hand side has maximum rank $2n_{A_1}$ thanks to the strict inequality $V_{A_1}>i\Omega$ (see the forthcoming Lemma~\ref{lemma gamma sharp}), we conclude that $2n_{E_1}\geq \rk \left( L \gamma_{E_1}^{-1} L^\intercal \right) = 2 n_{A_1}$, and hence $n_{E_1}\geq n_{A_1}$. But the same reasoning can be applied with $A_1$ and $E_1$ exchanged, thus giving $n_{A_1}\geq n_{E_1}$, which concludes the proof.
\end{proof}


If one wants to use Gaussian states to transmit and manipulate quantum
information, the role of measurements is of course central. We remind the reader
that a measurement in quantum theory is represented by a \emph{positive operator-valued measure} (POVM)  $E(dx)$ over a measurable
space $X$. Performing this measurement on a quantum state with
density matrix $\rho$ yields an outcome in $X$ according to the probability
distribution $p(dx)=\Tr[\rho E(dx)]$~\cite{Holevo book}.
Therefore, it is of prime importance for us to understand how
Gaussian states behave under measurements. Of course, the most natural
and easily implementable measurements are Gaussian as well, meaning that
the $X=\mathds{R}^{2n}$ and the positive operators $E(d^{2n}x)=E(x) d^{2n}x$ are positive
multiples of Gaussian states with a fixed covariance matrix $\sigma$ and
varying first moments $\Tr[E(x) r]\propto x$.
Implementing such a Gaussian measurement on a Gaussian state $\rho$
with a vector of first moments $s$ and a QCM $V$ yields
an outcome $x$ distributed according to a
Gaussian probability distribution
\begin{equation}
p(x) = \frac{2^n e^{-(x-s)^\intercal (V+\gamma)^{-1} (x-s)}}{\sqrt{\det(V+\sigma)}} .
\label{G meas}
\end{equation}
Furthermore, it can be shown that if a bipartite system $AB$ is in a Gaussian
state $\rho_{AB}$ described by a QCM $V_{AB}$ and only the second subsystem
$B$ is subjected to a Gaussian measurement described by a seed QCM $\sigma_B$, the
state of subsystem $A$ after the measurement, given by
$\rho'_A \propto \Tr_B [ \rho_{AB} \left( \text{id}_A\otimes E_B(x)\right) ]$, is again
Gaussian, and described by first moments depending on the measurement outcome,
but by a fixed QCM, given by the Schur complement~\cite{nogo1,nogo2,nogo3}
\begin{equation}
V'_B = (V_{AB} + 0_A\oplus\sigma_B) / (V_B +\sigma_B) .
\label{QCM after meas}
\end{equation}

Equation~\eqref{G meas} shows how quantum Gaussian states reproduce classical Gaussian probability
distributions when measured with Gaussian measurements. Thus, thanks
to the connection outlined in Section~\ref{sec:intro}, log-det entropies
become relevant in the quantum case as well, since they reproduce Shannon
entropies of the experimentally accessible measurement outcomes.
One could also wonder, whether the log-det entropy given in~\eqref{logdetent}
can be interpreted directly at the density operator level. To understand how this can be done, let us recall the notion of \emph{quantum R\'enyi-$\alpha$ entropy}
of a state $\rho$, given by
\begin{equation}
S_\alpha(\rho) \coloneqq \frac{1}{1-\alpha} \ln \Tr[ \rho^{\alpha}] .
\label{Renyi ent}
\end{equation}
Interestingly, it can be shown that for an arbitrary Gaussian state with QCM $V$
it holds
\begin{equation}
S_2(\rho) = \frac12 \ln \det V = M(V) = h(\xi) - n (\ln \pi  + 1),
\label{Renyi-2 G}
\end{equation}
i.e.~the \mbox{R\'enyi-$2$} entropy {\em coincides} with the log-det entropy
defined in~\eqref{logdetent}~\cite{Adesso12}, and  these quantities in turn coincide, up to an additive constant, with the differential entropy $h(\xi)$ of the classical Gaussian variable $\xi \in \mathds{R}^{2n}$ whose probability distribution is precisely the Wigner function $W_\rho (\xi)$ of the quantum Gaussian state $\rho$. In fact, \mbox{R\'enyi-$2$} quantifiers have repeatedly been shown to be useful in quantum optics, the underlying reason being that Gaussian states are particularly well-behaved when measures respecting their quadratic nature are employed~\cite{weedbrook12}.

Note that in general it is not advisable to form entropy expressions from R\'enyi entropies, since they do not obey any nontrivial constraints in a general multi-partite system~\cite{LMW}. In information theory, this is addressed by defining directly well-behaved notions of conditional R\'enyi entropy and R\'enyi mutual information~\cite{Tomamichel-book}. Here, we evade those issues as we are restricting to Gaussian states. In fact, as discussed in Section~\ref{sec:intro}, thanks to their special structure Gaussian states satisfy also \mbox{R\'enyi-$2$} entropic inequalities.
Not surprisingly, such inequalities find several applications in continuous variable
quantum information, in particular limiting the performances of quantum
protocols with Gaussian states.
For example, as demonstrated in~\cite{Adesso,Kor}, there is no Gaussian
state of a $(n_A+n_B+n_C)$-mode system $ABC$ that is simultaneously $A\rightarrow C$
and $B\rightarrow C$ steerable by Gaussian measurements when $n_C=1$. At the level of QCMs, this is a consequence
of the (non-balanced) inequality
\begin{equation}
M(V_{AC})+ M(V_{BC}) - M(V_A)-M(V_B)\geq 0,
\label{SSA purified}
\end{equation}
to be obeyed by all tripartite QCMs $V_{ABC}$. We stress that
\eqref{SSA purified} can not hold for all positive definite $V$ (that is, for all
classical covariance matrices), as it can be easily seen by
rescaling it via $V\mapsto kV$, for $k>0$. However, the new matrix
$V$ becomes unphysical for sufficiently small $k$, as it violates
the uncertainty principle~\eqref{Heisenberg}.


\subsection{Applications to SSA and entanglement quantification}

We are now ready to apply our results to strengthening the SSA inequality~\eqref{SSA} in the quantum case.
This subsection is thus devoted to finding
a sensible lower bound on the log-det conditional mutual information for
all QCMs. This bound will be given by a quantity called
\emph{\mbox{R\'enyi-$2$} Gaussian entanglement of formation}, already introduced
and studied in~\cite{Adesso12}. In general, for a bipartite quantum state $\rho_{AB}$, the \emph{R\'enyi-$\alpha$
entanglement of formation} is defined as the convex hull of the
R\'enyi-$\alpha$ entropy of entanglement defined on pure states
\cite{horodecki2009quantum}, i.e.
\begin{equation}
\begin{split}
  E_{F,\alpha}(A:B)_{\rho}
    &\coloneqq \inf \sum_i p_i \, S_{\alpha}\bigl(\psi_i^A\bigr) \\
           &\quad\,\text{ s.t. } \rho_{AB} = \sum_i p_i \psi_i^{AB} ,
\end{split}
  \label{EoF}
\end{equation}
where $\psi_{i}^{AB}$ are density matrices of pure states,
$\psi_i^A = \Tr_B \psi_i^{AB}$ is the reduced state (marginal), and $S_\alpha$
is defined in~\eqref{Renyi ent}.

For quantum Gaussian states, an upper bound to this quantity can be
derived by restricting the decompositions appearing in the above infimum
to be comprised of pure Gaussian states only. One obtains what is called
{\em Gaussian R\'enyi-$\alpha$ entanglement of formation}, a monotone under Gaussian local operations and classical communication, that in terms of
the QCM $V_{AB}$ of $\rho_{AB}$ is given by the simpler formula~\cite{Wolf03}
\begin{equation}
\begin{split}
  E^{\text{G}}_{F,\alpha}(A:B)_{V}
    &= \inf S_{\alpha}(\gamma_{A}) \\
    &\quad \text{ s.t. } \gamma_{AB} \text{ pure QCM and } \gamma_{AB}\leq V_{AB},
\end{split}
  \label{GEoF}
\end{equation}
where with a slight abuse of notation we denoted with $S_{\alpha}(W)$
the R\'enyi-$\alpha$ entropy of a Gaussian state with QCM $W$,
and $\gamma_{AB}$ stands for the QCM of a pure Gaussian state, i.e.~with $\det\gamma_{AB}=1$.
Incidentally, it has been proven~\cite{EoF symmetric G,Giovadd} that for some $2$-mode
Gaussian states the formula~\eqref{GEoF} reproduces exactly
\eqref{EoF}, i.e.~Gaussian decompositions in~\eqref{EoF} are globally optimal.

The most commonly used $E_{F,\alpha}$ is the one corresponding to
the von Neumann entropy, $\alpha=1$. However, as we already saw,
\mbox{R\'enyi-$2$} quantifiers arise quite naturally in the Gaussian setting, because by virtue of~\eqref{Renyi-2 G} they reproduce Shannon entropies of measurement outcomes, cf.~(\ref{G meas}). Thus, from now on
we will focus on the case $\alpha=2$. Under this assumption,
thanks to~\eqref{Renyi-2 G} we see that~\eqref{GEoF} becomes
\begin{equation}
\begin{split}
  E^{\text{G}}_{F,2}(A:B)_{V}
    &= \inf M(\gamma_{A}) \\[0.8ex]
    &\quad \text{ s.t. } \gamma_{AB} \text{ pure QCM and } \gamma_{AB}\leq V_{AB} .
\end{split}
  \label{G R2 EoF}
\end{equation}
We will find it convenient to rewrite the above equation in a slightly different form.
Using the well-known fact that $M(\gamma_A)=M(\gamma_B)=\frac12 I_M(A:B)_\gamma$
when $\gamma_{AB}$ is the QCM of a pure state~\cite{Adesso14}, we obtain
\begin{equation}
\begin{split}
  E^{\text{G}}_{F,2}(A:B)_{V}
    &= \inf \frac12 I_M(A:B)_{\gamma} \\[0.8ex]
    &\quad \text{ s.t. } \gamma_{AB} \text{ pure QCM and } \gamma_{AB}\leq V_{AB} .
\end{split}
  \label{G R2 EoF alt}
\end{equation}

The entanglement measure~\eqref{GEoF} is known to be faithful on quantum Gaussian states,
i.e.~it becomes zero if and only if the Gaussian state with QCM $V_{AB}$ is separable. Furthermore, in~\cite{Lami16} it was proven that
the Gaussian \mbox{R\'enyi-$2$} entanglement of formation obeys the notable inequality
\begin{equation}
  E^{\text{G}}_{F,2}(A:B)_V \leq \frac12 I_{M}(A:B)_V ,
  \label{noi}
\end{equation}
that in turn allows to prove useful {\em monogamy} properties of~\eqref{G R2 EoF}, captured by the inequality
\begin{equation}
  E^{\text{G}}_{F,2}(A:B_1\ldots B_n)_V \geq  \sum_{j=1}^n E^{\text{G}}_{F,2}(A:B_j)_V,
  \label{CKW}
  \end{equation}
  for any multipartite Gaussian state with QCM $V_{AB_1 \ldots B_n}$.




We are now in position to apply some of the tools we have been developing
so far to prove a generalisation of the inequality~\eqref{noi} that
is of interest to us since it constitutes also a strengthening of~\eqref{SSA}. Before coming to the main result of this subsection, we remind the reader of a useful result that extends~\cite[Lemma 13 (Supplemental material)]{Lami16}. Besides being a versatile tool to be employed throughout the rest of this section, the following lemma starts to show how fruitful the application of matrix analysis tools in quantum optics can be.

\begin{lemma} \label{lemma gamma sharp}
Let $K>0$ be a positive matrix. Then $\gamma_K^\# \equiv K\#(\Omega K^{-1}\Omega^\intercal)$ is a pure QCM. Furthermore, $K> i\Omega$ if and only if $K> \Omega K^{-1} \Omega^\intercal$, if and only if $K>\gamma_K^\#$.
\end{lemma}

\begin{proof}
We can follow the same steps as in the proof of~\cite[Lemma 13 (Supplemental material)]{Lami16}. Namely, we apply Lemma~\ref{lemma Williamson} to decompose $K=S \Delta S^T$, where $S$ is symplectic and $\Delta$ diagonal. Then, we deduce that
\begin{align*}
\gamma_K^\# &= (S\Delta S^\intercal) \# \left( \Omega S^{-\intercal} \Delta^{-1} S^{-1} \Omega^\intercal \right) \\
&\texteq{(i)} (S\Delta S^\intercal) \# \left( S \Omega \Delta^{-1} \Omega^\intercal S^\intercal \right) \\
&\texteq{(ii)} (S\Delta S^\intercal) \# \left( S \Delta^{-1} S^\intercal \right) \\
&\texteq{(iii)} S \left( \Delta \# \Delta^{-1} \right) S^\intercal \\
&\texteq{(iv)} SS^\intercal ,
\end{align*}
where we used, in order: (i) the identity $\Omega S^\intercal = S^{-1} \Omega$, valid for all symplectic $S$; (ii) the fact that $[\Omega,\Delta]=0$, which is a consequence of Lemma~\ref{lemma Williamson}; (iii) the congruence covariance of the geometric mean,~\eqref{geom cov congr}; and (iv) the elementary observation that $\Delta \# \Delta^{-1}=\mathds{1}$, as follows from the explicit formula~\eqref{geom expl}.
Then, it is easy to observe that $\gamma_{K}^\# $ is the QCM of a pure Gaussian state. The inequality $K>i\Omega$ translates to $\Delta > \mathds{1}$, and in turn to $K=S\Delta S^\intercal > SS^\intercal = \gamma_K^\#$, or alternatively to $\Delta>\Delta^{-1}$ and thus to $K = S \Delta S^\intercal > S \Delta^{-1} S^\intercal = \Omega K^{-1} \Omega^\intercal $. This latter condition can already be found in~\cite[Lemma 1]{sep 3-mode}.
\end{proof}

\begin{thm}
  \label{thm I cond G R2 EoF}
  For all tripartite QCMs $V_{ABC}\geq i\Omega_{ABC}$,
  it holds that
  \begin{equation}
    \frac12 I_{M}(A:B|C)_{V} \geq E^{\text{\emph{G}}}_{F,2}(A:B)_{V} .
    \label{ext I con}
  \end{equation}
\end{thm}

\begin{proof}
We employ a similar trick to the one used in~\cite{Lami16}: for any QCM $V_{ABC}$, using the notation of Lemma~\ref{lemma gamma sharp} define
\begin{equation}
  \gamma_{AB} \coloneqq \gamma^\#_{V_{ABC}/V_{C}} .
\end{equation}
Since $V_{ABC}/V_{C}>0$ by the positivity conditions~\eqref{Schur pos}, we see that $\gamma_{AB}$ is a pure QCM. Now we proceed to show that $\gamma_{AB}\leq V_{AB}$. On the one hand, the very definition of Schur complement implies that $V_{ABC}/V_{C}\leq V_{AB}$,
while on the other hand a special case of~\cite[Theorem 3]{Lami16} gives us the
general inequality $V_{ABC}/V_{C}\geq \Omega V_{AB}^{-1}\Omega^{\intercal}$,
i.e.~$\Omega (V_{ABC}/V_{C})^{-1} \Omega^{\intercal}\leq V_{AB}$.
Since the geometric mean is well-known to be monotonic~\cite{Ando79},
we obtain $\gamma_{AB}\leq V_{AB}$. This shows that $\gamma_{AB}$ can be used as an ansatz in~\eqref{G R2 EoF alt}.
We can write
\begin{align*}
     &E^{\text{G}}_{F,2}(A:B)_{V} \\[0.8ex]
     &\quad\leq \frac12 I_M(A:B)_\gamma \\[0.8ex]
     &\quad = \frac12 I_M(A:B)_{(V_{ABC}/V_{C})\# (\Omega (V_{ABC}/V_{C})^{-1}\Omega^\intercal ) } \\[0.8ex]
     &\quad\textleq{(i)} \frac14 I_M(A:B)_{V_{ABC}/V_{C}} + \frac14 I_M(A:B)_{\Omega (V_{ABC}/V_{C})^{-1}\Omega^\intercal} \\[0.8ex]
     &\quad\texteq{(ii)} \frac14 I_M(A:B)_{V_{ABC}/V_{C}} + \frac14 I_M(A:B)_{(V_{ABC}/V_{C})^{-1}} \\[0.8ex]
     &\quad\texteq{(iii)} \frac14 I_M(A:B|C)_V + \frac14 I_M(A:B|C)_V \\[0.8ex]
     &\quad= \frac12 I_M(A:B|C)_V ,
\end{align*}
where we employed, in order: (i) the convexity of log-det mutual information on the trace metric geodesics~\eqref{I conv geod}, (ii) the obvious fact that since $\Omega_{AB}=\Omega_A \oplus \Omega_B$, the equality $I_M(A:B)_{\Omega W \Omega^\intercal }=I_M(A:B)_W$ holds true; and (iii) the identity~\eqref{I cond Schur} for the first term and~\eqref{I cond inv} followed again by~\eqref{I cond Schur} for the second.
\end{proof}

\subsection{Gaussian \mbox{R\'enyi-$2$} squashed entanglement} \label{subsec Gauss sq}

In finite-dimensional quantum mechanics, the positivity of conditional
mutual information allows to construct a powerful entanglement measure
called \emph{squashed entanglement}, defined for a bipartite state
$\rho_{AB}$ by~\cite{CW04}
\begin{equation}
  E_{\text{sq}}(A:B)_{\rho} \coloneqq \inf_{\rho_{ABC}} \frac12 I(A:B|C)_{\rho} ,
\end{equation}
where the infimum ranges over all possible ancillary quantum systems $C$
and over all the possible states $\rho_{ABC}$ having marginal $\rho_{AB}$.
We are now in position to discuss a similar quantity tailored to Gaussian states.
First, we can restrict the infimum by considering only Gaussian extensions,
which corresponds to the step leading from~\eqref{EoF} to~\eqref{GEoF}.
Secondly, as it was done to arrive at~\eqref{G R2 EoF}, we can substitute
von Neumann entropies with \mbox{R\'enyi-$2$} entropies. The result is
\begin{equation}
  E^{\text{G}}_{\text{sq},2}(A:B)_{V} \coloneqq \inf_{V_{ABC}} \frac12 I_{M}(A:B|C)_{V},
  \label{Gauss sq}
\end{equation}
where the infimum is on all extended QCMs $V_{ABC}$ satisfying the condition
$\Pi_{AB} V_{ABC}\Pi_{AB}^\intercal = V_{AB}$ on the $AB$ marginal (and~\eqref{Heisenberg}).
We dub the quantity in~\eqref{Gauss sq} {\em Gaussian \mbox{R\'enyi-$2$} squashed
entanglement}, stressing that it is a quantifier specifically tailored to Gaussian states
and different from the R\'enyi squashed entanglement defined
in~\cite{Wilde} for general states, where an alternative expression
for the conditional R\'enyi-$\alpha$ mutual information is adopted instead.

Despite the complicated appearance of the expression~\eqref{Gauss sq}, it turns out that \emph{the Gaussian \mbox{R\'enyi-$2$} squashed entanglement coincides with the Gaussian \mbox{R\'enyi-$2$} entanglement of formation for all bipartite QCMs}. This unexpected fact shows once more that \mbox{R\'enyi-$2$} quantifiers are particularly well behaved when employed to analyse Gaussian states, while at the same time it provides us with a novel, alternative expression of $E^{\text{G}}_{F,2}$ that can be used to understand its basic properties in a different, and sometimes more intuitive, way. Before stating the main result of this subsection, we need some preliminary results.

\begin{lemma} \label{lemma follia 0}
Let $\gamma_{AB}$ be a pure QCM of a bipartite system $AB$ such that $n_A=n_B=n$ and $\gamma_A>i\Omega_A$. Then
\begin{equation*}
\left(\gamma_{AB}+ i \Omega_{AB}\right) \big/ \left( \gamma_A + i\Omega_A \right) = 0_B .
\end{equation*}
\end{lemma}

\begin{proof}
From Williamson's decomposition, Lemma~\ref{lemma Williamson}, we see that whenever $\gamma_{AB}$ is pure one has $\rk(\gamma_{AB}+i\Omega_{AB}) = n_A + n_B = 2n$ (i.e. half the maximum). Since already $\rk(\gamma_A + i\Omega_A) = 2n$, the additivity of ranks under Schur complements~\eqref{rank add} tells us that $\rk \left( \left(\gamma_{AB}+ i \Omega_{AB}\right) \big/ \left( \gamma_A + i\Omega_A \right) \right) = 0$, concluding the proof.
\end{proof}

\begin{prop} \label{follia prop}
Let $V_{AB}$ be a QCM of a bipartite system, and let $\gamma_{ABC}$ be a fixed purification of $V_{AB}$ (see Lemma~\ref{lemma pur}). Then, for all pure QCMs $\tau_{AB} \leq V_{AB}$ there exists a one-parameter family of pure QCMs $\sigma_C (t)$ (where $0<t\leq 1$) on $C$ such that
\begin{equation}
\gamma'_{AB}(t)\coloneqq \left(\gamma_{ABC} + 0_{AB} \oplus \sigma_C(t) \right) \big/ \left( \gamma_C +\sigma_C(t) \right)
\label{follia prop eq}
\end{equation}
is a pure QCM for all $t>0$, and $\lim_{t\rightarrow 0^+} \gamma'_{AB}(t) = \tau_{AB}$. Equivalently, there is a sequence of Gaussian measurements on $C$, identified by pure seeds $\sigma_C(t)$, such that the QCM of the post-measurement state on $AB$ is pure and tends to $\tau_{AB}$ (see~\eqref{QCM after meas}). 
\end{prop}

\begin{proof}
Let us start by applying Lemma~\ref{lemma fact out} to decompose the symplectic space of $AB$ as $\Sigma_{AB}=\Sigma_R \oplus \Sigma_S$ in such a way that $V_{AB}=V_R \oplus \eta_S$, where $V_R > i\Omega_R$ and $\eta_S$ is a pure QCM. According to Lemma~\ref{lemma fact out}, the purification $\gamma_{ABC}$ can be taken to be of the form $\gamma_{ABC}=\gamma_{RC_{1}} \oplus \eta_S\oplus \delta_{C_{2}}$, with $\gamma_{C_{1}}>i\Omega_{C_{1}}$, $n_{C_{1}}=n_R$, and $\delta_{C_{2}}$ pure. If $\tau\leq V$ is a pure QCM, a projection onto $\Sigma_S$ reveals that $\tau_S = \Pi_S \tau \Pi_S^\intercal \leq \eta_S$. Since $\tau_S$ must be a legitimate QCM, and pure states are minimal within the set of QCMs, we deduce that $\tau_S=\eta_S$. Then, an application of Lemma~\ref{lemma pure reduction} allows us to conclude that $\tau = \tau_R \oplus \eta_S$, and accordingly $\tau_R \leq V_R$.

We claim that for all pure $\tau_R < V_R$ there is a pure QCM $\sigma_{C_{1}}$ such that
\begin{equation}
\left(\gamma_{RC_{1}} + 0_{R} \oplus \sigma_{C_{1}} \right) \big/ \left( \gamma_{C_{1}} +\sigma_{C_{1}} \right) = \tau_R .
\label{follia prop eq2}
\end{equation}
Constructing the extension $\sigma_{C}\coloneqq \sigma_{C_{1}}\oplus \tilde{\sigma}_{C_{2}}$, where $\tilde{\sigma}_{C_{2}}$ is an arbitrary pure QCM, we see that~\eqref{follia prop eq2} can be rewritten as 
\begin{equation}
\left(\gamma_{ABC} + 0_{AB} \oplus \sigma_{C} \right) \big/ \left( \gamma_{C} +\sigma_{C} \right) = \tau_R \oplus \eta_{S} .
\label{follia prop eq2bis}
\end{equation}
In fact, adding the ancillary system $C_{2}$ does not produce any effect on the Schur complement, since there are no off-diagonal block linking $C_{2}$ with any other subsystem. Analogously, the $S$ component of the $AB$ system can be brought out of the Schur complement because it is in direct sum with the rest.

In light of~\eqref{follia prop eq2bis}, we know that once~\eqref{follia prop eq2} has been established, in~\eqref{follia prop eq} we can achieve all QCMs $\gamma'$ that can be written as $\tau_{R}\oplus \eta_{S}$, with $\tau_{R}<V_{R}$. It is not difficult to see that this would allow us to conclude. Before proving~\eqref{follia prop eq2}, let us see why. The main point here is that every pure QCM $\tau_R\leq V_R$ can be thought of as the limit of a sequence of pure QCMs $\tau_R(t)< V_R$. An explicit formula for such a sequence reads $\tau_R(t) = \tau_R \#_t \gamma_{V_R}^\#$, where $\gamma_{V_R}^\#$ is the pure QCM defined in Lemma~\ref{lemma gamma sharp}, and $\#_t$ denotes the weighted geometric mean~\eqref{geom geod}. Observe that: (i) $\tau_R(t)$ is a QCM since it is known that the set of QCMs is closed under weighted geometric mean~\cite[Corollary 8]{Bhatia symplectic eig}; (ii) $\tau_R(t)$ is in fact a pure QCM, because according to~\eqref{det geom} its determinant satisfies $\det \tau_R(t) = \left(\det \tau_R\right)^{1-t} \big(\det \gamma_{V_R}^\# \big)^t = 1$; (iii) $\lim_{t\rightarrow 0^+} \tau_R(t)=\tau_R$ as can be seen easily from~\eqref{geom geod}; and (iv) $\tau_R(t)< V_R$ for all $t>0$. This latter fact can be justified as follows. Since $V_R > i\Omega_R$, from Lemma~\ref{lemma gamma sharp} we deduce $\gamma_{V_R}^\# < V_R$. Taking into account that $\tau_R \leq V_R$, the claim follows from the strict monotonicity of the weighted geometric mean, in turn an easy consequence of~\eqref{geom geod}.

Now, let us prove~\eqref{follia prop eq2}. We start by writing
\begin{equation*}
\gamma_{RC_{1}} = \begin{pmatrix} V_R & L \\ L^\intercal & \gamma_{C_{1}} \end{pmatrix} ,
\end{equation*}
where $V_R>i\Omega_R$, $\gamma_{C_{1}}> i\Omega_{C_{1}}$, and the off-diagonal block $L$ is square. As a matter of fact, more is true, namely that $L$ is also invertible. The simplest way to see this involves two ingredients: (a) the identity $\Omega V_{R}^{-1} \Omega^\intercal = \gamma_{RC_{1}} / \gamma_{C_{1}} = V_{R} - L \gamma_{C_{1}}^{-1} L^\intercal$, easily seen to be a special case of~\cite[Equation (8)]{Lami16}; and (b) the fact that $V_R> \Omega V_R^{-1} \Omega^\intercal$ because of Lemma~\ref{lemma gamma sharp}. Combining these two ingredients we see that
\begin{equation*}
V_{R} > \Omega V_{R}^{-1} \Omega^\intercal = V_{R} - L \gamma_{C_{1}}^{-1} L^\intercal ,
\end{equation*}
which implies $L \gamma_{C_{1}}^{-1} L^\intercal>0$ and in turn the invertibility of $L$. Now, for a pure QCM $\tau_R< V_R$, take $\sigma_{C_{1}} = L^\intercal (V_R-\tau_R)^{-1} L - \gamma_{C_{1}}$. On the one hand, 
\begin{align*}
\big(\gamma_{RC_{1}} \!+ 0_R\!\oplus\! \sigma_{C_{1}}\big) \big/ \big(\gamma_{C_{1}} \!+ \sigma_{C_{1}}\big) &= V_R - L \left(\gamma_{C_{1}}\! + \sigma_{C_{1}}\right)^{-1}\! L^\intercal \\[0.8ex]
&= \tau_R
\end{align*}
by construction. On the other hand, write
\begin{align*}
\sigma_{C_{1}} \!- i\Omega_{C_{1}} &= L^\intercal (V_R - \tau_R)^{-1} L - (\gamma_{C_{1}} + i\Omega_{C_{1}}) \\[0.8ex]
&= L^\intercal (V_R - \tau_R)^{-1} L - L^\intercal (V_R + i \Omega_R)^{-1} L \\[0.8ex]
&= L^\intercal \left( (V_R - \tau_R)^{-1} - (V_R + i \Omega_R)^{-1} \right) L \\[0.8ex]
&= L^\intercal (V_R - \tau_R)^{-1} \\
&\quad \cdot \left( (V_R + i \Omega_R) - (V_R - \tau_R) \right) \\
&\quad \cdot (V_R + i \Omega_R)^{-1} L \\[0.8ex]
&= L^\intercal (V_R - \tau_R)^{-1} \left( \tau_R + i \Omega_R \right) (V_R + i \Omega_R)^{-1} L ,
\end{align*}
where we employed Lemma~\ref{lemma follia 0} in the form $\gamma_{C_{1}} + i \Omega_{C_{1}} = L^\intercal (V_R + i \Omega_R)^{-1} L$ and performed some elementary algebraic manipulations. Now, from the third line of the above calculation it is clear that $\sigma_{C_{1}} - i \Omega_{C_{1}}\geq 0$, since from $V_R - i \Omega_R \geq V_R - \tau_R > 0$ we immediately deduce $(V_R - \tau_R)^{-1} \geq (V_R + i \Omega_R)^{-1}$. This shows that $\sigma_{C_{1}}$ is a valid QCM. Moreover, observe that
\begin{align*}
&\rk \left(\sigma_{C_{1}} - i \Omega_{C_{1}}\right) \\[0.8ex]
&\quad = \rk \left( L^\intercal (V_R - \tau_R)^{-1} \left( \tau_R + i \Omega_R \right) (V_R + i \Omega_R)^{-1} L \right) \\[0.8ex]
&\quad= \rk \left( \tau_R + i \Omega_R \right) \\
&\quad= n_R \\[0.8ex]
&\quad= n_{C_{1}} ,
\end{align*}
which tells us that $\sigma_{C_{1}}$ is also a pure QCM.
\end{proof}

Now, we are ready to state the conclusive result of the present paper.

\begin{thm}
  \label{thm GSq=GEoF}
  For all bipartite QCMs $V_{AB}\geq i\Omega_{AB}$, the Gaussian \mbox{R\'enyi-$2$} squashed
  entanglement coincides with the Gaussian \mbox{R\'enyi-$2$} entanglement of formation, i.e.
  \begin{equation}
  E^{\text{\emph{G}}}_{\text{\emph{sq}},2}(A:B)_{V} = E_{F,2}^{\text{\emph{G}}}(A:B)_V .
  \label{GSq=GEoF}
  \end{equation}
\end{thm}

\begin{proof}
The inequality $E^{\text{G}}_{\text{sq},2}(A:B)_V \geq E_{F,2}^{\text{G}}(A:B)_V$ is an easy consequence
of~\eqref{ext I con} together with~\eqref{Gauss sq}. To show the converse, we employ the expression~\eqref{G R2 EoF alt} for the Gaussian R\'enyi-2 entanglement of formation. Consider an arbitrary purification $\gamma_{ABC}$ of $V_{AB}$, and pick a pure state $\tau_{AB}\leq V_{AB}$. By construction, we have $\gamma_{AB}=V_{AB}$. Now, thanks to Proposition~\ref{follia prop} one can construct a sequence of measurements identified by $\sigma_C(t)$ such that~\eqref{follia prop eq2} holds. Then, we have
\begin{align*}
&\frac12 I_M(A:B)_\tau \\[0.8ex]
&\quad= \frac12 I_M(A:B)_{\lim_{t\rightarrow 0^+} \left(\gamma_{ABC} + 0_{AB} \oplus \sigma_C(t) \right) / \left( \gamma_C +\sigma_C(t) \right) } \\
&\quad\texteq{(i)} \lim_{t\rightarrow 0^+} \frac12 I_M(A:B)_{\left(\gamma_{ABC} + 0_{AB} \oplus \sigma_C(t) \right) / \left( \gamma_C +\sigma_C(t) \right) } \\
&\quad\texteq{(ii)} \lim_{t\rightarrow 0^+} \frac12 I_M(A:B|C)_{\gamma_{ABC} + 0_{AB} \oplus \sigma_C(t)} \\
&\quad\textgeq{(iii)} E^{\text{G}}_{\text{sq},2}(A:B)_V ,
\end{align*}
where we used, in order: (i) the continuity of the log-det mutual information; (ii) the identity~\eqref{I cond Schur}; and (iii) the fact that the QCMs $\gamma_{ABC} + 0_{AB} \oplus \sigma_C(t)$ constitute valid extensions of $V_{AB}$, thus being legitimate ansatzes in~\eqref{Gauss sq}.
\end{proof}

\begin{rem}
A by-product of the above proof of Theorem~\ref{thm GSq=GEoF} is that in~\eqref{Gauss sq} we can restrict ourselves to systems of bounded size $n_C \leq n_{AB} = n_A + n_B$. Moreover, up to limits the extension can be taken of the form $\gamma_{ABC} + 0_{AB} \oplus \sigma_C$, where $\gamma_{ABC}$ is a fixed purification of $V_{AB}$ and $\sigma_C$ is a pure QCM.
\end{rem}

This surprising identity between two seemingly very different entanglement measures, even though tailored to Gaussian states, is remarkable. On the one hand, it provides an interesting operational interpretation for the Gaussian \mbox{R\'enyi-$2$} entanglement of formation in terms of log-det conditional mutual information, via the recoverability framework. On the other hand, it simplifies the notoriously difficult evaluation of the squashed entanglement, in this case restricted to Gaussian extensions and log-det entropy, because it recasts it as an optimisation of the form~\eqref{G R2 EoF}, which thus involves matrices of bounded instead of unbounded size (more precisely, of the same size as the mixed QCM whose entanglement is being computed).
In general, Theorem~\ref{thm GSq=GEoF} allows us to export useful properties between the two frameworks it connects. For instance, it follows from the identity~\eqref{GSq=GEoF} that the Gaussian \mbox{R\'enyi-$2$} squashed entanglement is faithful on Gaussian states and a monotone under Gaussian local operations and classical communication; in contrast, proving the property of faithfulness for the standard squashed entanglement was a very difficult step to perform~\cite{Brandao11}. On the other hand, the arguments establishing many basic properties of the standard squashed entanglement can be imported from~\cite{CW04} and applied to~\eqref{Gauss sq}, providing new proofs of the same properties for the Gaussian \mbox{R\'enyi-$2$} entanglement of formation. Let us give an example of how effective is the interplay between the two frameworks by providing an alternative, one-line proof of the following result~\cite{Lami16}

\begin{lemma}~\cite[Corollary 7]{Lami16}
The Gaussian \mbox{R\'enyi-$2$} entanglement of formation is monogamous on arbitrary Gaussian states, i.e.
\begin{equation}
E_{F,2}^{\text{\emph{G}}}(A:BC) \geq E_{F,2}^{\text{\emph{G}}}(A:B) + E_{F,2}^{\text{\emph{G}}}(A:C) ,
\label{monogamy GEoF}
\end{equation}
and analogously for more than three parties.
\end{lemma}

\begin{proof}
Thanks to Theorem~\ref{thm GSq=GEoF}, we can prove the monogamy relation~\eqref{monogamy GEoF} for the Gaussian \mbox{R\'enyi-$2$} squashed entanglement instead. We use basically the same argument as in~\cite[Proposition 4]{CW04}. Namely, call $V_{ABC}$ the QCM of the system $ABC$. Then for all extensions $V_{ABCE}$ of $V_{ABC}$ one has
\begin{align*}
I_M(A:BC|E)_V &= I_M(A:B|E)_V + I_M (A:C|BE) \\[0.8ex]
&\geq 2  E^{\text{G}}_{\text{sq},2}(A:B)_V + 2 E^{\text{G}}_{\text{sq},2}(A:C)_V ,
\end{align*}
where we applied the chain rule for the conditional mutual information together with the obvious facts that $V_{ABE}$ is a valid extension of $V_{AB}$ and $V_{ABCE}$ a valid extension of $V_{AC}$.
\end{proof}

A monogamy inequality is a powerful tool in dealing with entanglement measures. For instance, when combined with monotonicity under local operations, it leads to the additivity of the measure under examination.

\begin{cor}
The Gaussian entanglement measure $E_{F,2}^{\text{\emph{G}}} = E^{\text{G}}_{\text{sq},2}$ is additive under tensor products (equivalently, direct sum of covariance matrices). In formulae,
\begin{equation}
\begin{split}
E_{F,2}^{\text{\emph{G}}} (A_1 A_2: B_1 B_2)_{V_{A_1 B_1} \oplus W_{A_2 B_2}} &= E_{F,2}^{\text{\emph{G}}} (A_1: B_1)_{V} \\
&\quad + E_{F,2}^{\text{\emph{G}}} (A_2: B_2)_{W} .
\end{split}
\label{additivity GEoF}
\end{equation}
\end{cor}

\begin{proof}
Applying first~\eqref{monogamy GEoF} and then the monotonicity of $E_{F,2}^{\text{\emph{G}}}$ under the operation of discarding some local subsystems, we obtain
\begin{align*}
&E_{F,2}^{\text{\emph{G}}} (A_1 A_2: B_1 B_2)_{V_{A_1 B_1} \oplus W_{A_2 B_2}} \\[0.8ex]
&\quad \geq E_{F,2}^{\text{\emph{G}}} (A_1 A_2 : B_1)_{V_{A_1 B_1}\oplus W_{A_2}} \\
&\quad\quad + E_{F,2}^{\text{\emph{G}}} (A_1 A_2: B_2)_{V_{A_1} \oplus W_{A_2 B_2}} \\[0.8ex]
&\quad \geq E_{F,2}^{\text{\emph{G}}} (A_1: B_1)_{V} + E_{F,2}^{\text{\emph{G}}} (A_2: B_2)_{W} .
\end{align*}
The opposite inequality follows by inserting factorised ansatzes $\gamma_{A_1 B_1}\oplus \tau_{A_2 B_2}$ into~\eqref{GEoF}.
\end{proof}

As established in this section, the Gaussian \mbox{R\'enyi-$2$} entanglement of formation alias Gaussian \mbox{R\'enyi-$2$} squashed entanglement also emerges as a rare example of an additive entanglement monotone (within the Gaussian framework) which satisfies the general monogamy inequality (\ref{CKW}). We remark that the conventional (R\'enyi-$1$) entanglement of formation cannot fundamentally be monogamous~\cite{Lancien}, while the standard squashed entanglement is monogamous on arbitrary multipartite systems~\cite{Koashi}.

\section{Conclusions}
\label{sec:conchi}
In this paper, we analysed the SSA inequality for the log-det entropy from a matrix analysis viewpoint and explored some of its applications. We first derived new necessary and sufficient conditions for saturation of said inequality. In the context of classical recoverability, we then provided an explicit form for the  Gaussian Petz recovery map and further obtained a strengthening of SSA by constructing a faithful lower bound to a log-det entropy based conditional mutual information. We finally specialised to quantum Gaussian states, for which the log-det entropy reduces to the \mbox{R\'enyi-$2$} entropy, and defined a corresponding Gaussian version of the squashed entanglement measure. Surprisingly, we showed that the latter measure coincides with the \mbox{R\'enyi-$2$} entanglement of formation defined via a Gaussian convex roof construction~\cite{Adesso12}.
In turn, this allows us to build a bridge connecting the two frameworks, that can be used to establish new properties of a measure by looking at the other, or to provide simpler and more instructive proofs of known properties.

This manuscript, following a recent series of contributions~\cite{Adesso12,Adesso,Lami16}, casts further light on the connections between matrix analysis (in particular determinantal inequalities) and information theory in both classical and quantum settings. In future work, within the context of continuous variable quantum information with Gaussian states~\cite{Adesso14}, it could be interesting to establish whether the equivalence between the Gaussian \mbox{R\'enyi-$2$} squashed entanglement defined here and the Gaussian \mbox{R\'enyi-$2$} entanglement of formation defined in~\cite{Adesso12} further extends to a third measure of entanglement, namely the recently introduced Gaussian intrinsic entanglement~\cite{GIE}. It could also be worth exploring whether, for states where Gaussian decompositions attain the global convex roof optimisation for the entanglement of formation (such as symmetric $2$-mode Gaussian states), one could extend our techniques to show that even the standard squashed entanglement defined in terms of von Neumann conditional mutual information~\cite{CW04} may be optimised by Gaussian extensions and perhaps coincide with the conventional entanglement of formation; this would constitute a unique instance of computable squashed entanglement on states very relevant for applications in quantum optics. Finally, while we have studied the classical Gaussian Petz recovery map here, a very recent study has investigated the quantum Petz map for Gaussian states, showing it to be a Gaussian channel~\cite{LL17}. The results and techniques in~\cite{LL17} are however of somewhat different nature than those presented here, since they involve quantum states instead of classical random variables. In this context, it will be interesting to investigate Gaussian measures of more general quantum correlations~\cite{Adesso16} based on the fidelity of recovery after Gaussian entanglement-breaking channels, in analogy to the finite-dimensional case~\cite{Wilde2}.

\ifCLASSOPTIONcaptionsoff
  \newpage
\fi



%

%

\begin{IEEEbiographynophoto}{Ludovico Lami} received a B.Sc. degree in Physics in 2012 and a M.Sc. degree in Theoretical Physics in 2014, both from the Universit\`a di Pisa, Pisa, Italy. In 2015 he received also a Diploma in Physics at the Scuola Normale Superiore, Pisa, Italy. He is currently a Ph.D. student of Andreas Winter at the Universitat Aut\`onoma de Barcelona, Barcelona, Spain. His research interests lie primarily in quantum information, especially with continuous variable, but include also foundational aspects of quantum physics.
\end{IEEEbiographynophoto}

\begin{IEEEbiographynophoto}{Christoph Hirche} received a B.Sc. degree in Physics in 2013 and a M.Sc. degree in Physics in 2015, both from the Leibniz Universit\"{a}t Hannover, Hannover, Germany. He is currently a Ph.D. student at the Universitat Aut\`{o}noma de Barcelona, Barcelona, Spain. His research interests are mostly in mathematical aspects of quantum information theory.
\end{IEEEbiographynophoto}


\begin{IEEEbiographynophoto}{Gerardo Adesso} received a M.Sc. degree in Physics from the University of Salerno, Salerno, Italy, in 2003, and a Ph.D. degree in Physics from the same University in 2007, having spent one third of his Ph.D. in the Centre for Quantum Computation at the University of Cambridge, Cambridge, UK. After post-doctoral positions at Sapienza University of Rome, Rome, Italy, and Universitat Aut\`onoma de Barcelona, Barcelona, Spain, he joined the University of Nottingham, Nottingham, UK, where he was appointed Lecturer in 2009, Associate Professor in 2014, and Professor of Mathematical Physics in 2016. He was awarded an ERC Starting Grant in 2015 and was honoured as Young Scientist by the World Economic Forum in 2016. His research interests include continuous variable quantum information theory and the characterisation of nonclassical resources such as quantum coherence and correlations.\end{IEEEbiographynophoto}

\begin{IEEEbiographynophoto}{Andreas Winter} received a Diploma degree in Mathematics from the Freie
Universit\"at Berlin, Berlin, Germany, in 1997, and a Ph.D. degree from 
the Fakult\"at f\"ur Mathematik, Universit\"at Bielefeld, Bielefeld, Germany, 
in 1999. He was Research Associate at the University of Bielefeld until 
2001, and then with the Department of Computer Science at the University 
of Bristol, Bristol, UK. In 2003, still with the University of Bristol, 
he was appointed Lecturer in Mathematics, and in 2006 Professor of 
Physics of Information. Since 2012 he has been ICREA Research Professor 
with the Universitat Aut\`onoma de Barcelona, Barcelona, Spain. His research 
interests include quantum and classical Shannon theory, and discrete 
mathematics. He is recipient, along with Bennett, Devetak, Harrow and 
Shor, of the 2017 Information Theory Society Paper Award.
\end{IEEEbiographynophoto}




\end{document}